\documentclass[11pt]{article}
\usepackage{amsmath}
\usepackage{amssymb}
\usepackage{amsthm}
\usepackage{authblk}

\usepackage[final]{hyperref}
\usepackage[x11names]{xcolor}
\usepackage[modulo]{lineno}

\usepackage{fullpage}
\usepackage{framed}
\usepackage{amsfonts, latexsym, color, paralist, cancel,ifdraft}
\usepackage[final]{graphicx}
\ifdraft{\usepackage{showlabels}}{}
\ifdraft{\linenumbers}{}
\usepackage[refpage]{nomencl}

\hypersetup{
  colorlinks   = true, 
  urlcolor     = blue, 
  linkcolor    = {blue!55!black}, 
  citecolor   = red 
}

\advance \topmargin by -\headheight
\advance \topmargin by -\headsep
\evensidemargin \oddsidemargin
\newcommand\numberthis{\addtocounter{equation}{1}\tag{\theequation}}

\newcommand{\abs}[1]{\left\vert #1 \right\vert}
\newcommand{\alg}{\mathcal{A}}
\newcommand{\bra}[1]{\left\langle #1 \right\vert}

\newcommand{\floor}[1]{\lfloor #1 \rfloor}
\newcommand{\func}{f}
\newcommand{\ham}{\mathcal{H}}

\newcommand{\ket}[1]{\left\vert #1 \right\rangle}

\newcommand{\lang}{\mathcal{L}}
\newcommand{\norm}[1]{\left\| #1 \right\|}
\newcommand{\poly}[1]{\mathrm{poly(} #1 \mathrm{)}}
\newcommand{\prob}[1]{\Pr[ #1 ]}
\newcommand{\tr}[1]{tr\left( #1 \right)}

\newcommand{\refeq}[1]{(\ref{#1})}
\newcommand{\pBQP}{\mathrm{PostBQP}}
\newcommand{\BQP}{\mathrm{BQP}}
\newcommand{\BPP}{\mathrm{BPP}}
\newcommand{\IP}{\mathrm{IP}}
\newcommand{\PP}{\mathrm{PP}}
\newcommand{\MIP}{\mathrm{MIP}}
\newcommand{\NEXP}{\mathrm{NEXP}}
\newcommand{\PSPACE}{\mathrm{PSPACE}}
\newcommand{\p}{\mathrm{P}}
\newcommand{\QMA}{\mathrm{QMA}}
\newcommand{\PsP}{\p^{\#\p}}
\newcommand{\sP}{\#\p}
\newcommand{\PPP}{\p^{\PP}}
\newcommand{\PpBQP}{\p^{\pBQP}}

\theoremstyle{plain}
\newtheorem{theorem}{Theorem}[section]
\newtheorem{theorem*}{Theorem}
\newtheorem{lemma}{Lemma}
\newtheorem{claim}{Claim}
\newtheorem{subclaim}{Claim}[claim]

\theoremstyle{definition}
\newtheorem{defn}{Definition}[section]

\theoremstyle{remark}

\begin{document}

\title{A Quantum inspired proof of $\PsP \subseteq \IP$}

\author{Dorit Aharonov, Ayal Green} 
\affil{School of Computer Science and Engineering, The Hebrew University, Jerusalem 91904,  Israel}
\maketitle

\abstract{We provide a new, quantum inspired, proof for the celebrated 
claim of \cite{LFKN92} that $\PsP\subseteq \IP$. The protocol 
is fundamentally different from the original sum-check 
protocol of \cite{LFKN92, Sha92}, as well as from variants of this proof 
\cite{GKR08, Mei09}, 
though it still possesses the overall structure of inductively checking 
consistency between subsequent step.   
The protocol is inspired by \cite{AAV13}. 
Hopefully, this protocol will be helpful in 
making progress towards two major open problems: resolving the 
quantum PCP question, 
and verification of quantum computations using a classical $\BPP$ 
verifier. Given the historical importance of
the sum-check protocol in classical computational complexity, 
we hope that the new protocol will also find applications in 
classical complexity theory. 

\section{Introduction}
The \textit{sum-check} protocol, 
introduced in the celebrated $\PsP \subseteq \IP$ and the $\IP=\PSPACE$ 
results \cite{LFKN92,Sha92},  
played a key role 
in computational complexity, most notably 
leading to the proof that $\MIP=\NEXP$ \cite{BFL91}
and then to the (original, algebraic) proof of the celebrated PCP 
Theorem \cite{ALM+98, AS98}.  
Another important example application 
comes from cryptography: secure delegation of 
computations (initiated in \cite{GKR08}), in which the honest server is a 
$\BPP$ machine, but the user (verifier) is much weaker computationally. 

Deriving quantum versions of both of these major results remains open. 
The quantum PCP conjecture \cite{AALV09, AAV13}
remains unresolved despite many recent 
advancements (in both directions), see  
e.g. \cite{BH13, FH13, EH15, FV15, Ji16a, Ji16b, NV16}. 
It is considered a major problem in the field of Hamiltoniam 
Complexity. Interestingly, yet another holy grail in quantum complexity 
has to do with the delegation of quantum computation.   
Here the goal is slightly different than its 
classical counterpart. The question is whether a $\BPP$ verifier 
can delegate polynomial time quantum computations (namely 
computations in $\BQP$) to an honest $\BQP$ prover (like in the classical case,
security is required against dishonest provers with unbounded 
computational power). 
The initial protocols for the problem \cite{ABOE10, FK12}
allowed the $\BPP$ verifier to process constantly many quantum bits; 
However despite much follow-up work on related questions (see, e.g., 
\cite{RUV13, FM16, ABOEM17} it is still open whether verification 
can be done with a single prover, when the verifier is entirely classical, 
namely a $\BPP$ machine.  
 
The fact that these two central questions remain open 
despite increasing interest and research, and although the 
respective classical results are well known, is no coincidence. 
There seem to be some inherent obstacles towards 
quantum generalizations of classical interactive proof techniques (see, e,g, 
\cite{AALV09, AAV13}). In particular, these tools 
rely heavily on {\it local} tests, whereas the correlations in  
quantum states are highly non-local and cannot be addressed locally.   
We are thus motivated to develop alternative tools for 
interactive protocols, which are more natural to 
the quantum setting. In this work we revisit the fundamental 
sum-check protocol of \cite{LFKN92, Sha92}, 
and provide an alternative protocol for the
 task; we believe the new protocol is better adapted to the quantum world, and in particular possesses 
an inherent tensor product structure.   
We stress that the resulting theorem itself, namely 
a rederivation of \cite{LFKN92}, is entirely classical, and in particular   
neither the verifier nor the prover in the protocol need be quantum.

\subsection{Our Results} 
For two complexity classes  $\mathcal{V}, \mathcal{P}$, 
denote the complexity class $\IP[\mathcal{V}, \mathcal{P}]$ to be (informally)
the set of all languages for which there exists an interactive proof
protocol between a verifier which has computational power $\mathcal{V}$ and a prover, such that the protocol is complete even if the prover is 
computationally limited to $\mathcal{P}$, but where the soundness holds regardless of any computational assumption on the prover. 

Our basic protocol works in the quantum setting; it 
works for problems in $\BQP$ (in fact, in $\pBQP$. Recall that 
roughly, this class is the class of problems which can be 
solved by quantum computers which are allowed to project 
some of the qubits on desired results; see Section \ref{sec:PostBQP}
for exact definition). The verifier is a $\BPP$ machine, 
and the honest prover is required to be $\pBQP$. 
We derive an interactive protocol which proves: 
\begin{theorem}\label{thm:maintheorem1} $\pBQP\in \IP[\BPP, \pBQP]$ \label{thm:thm1a}. Furthermore, the $\IP[\BPP, \pBQP]$ protocol has completeness 1.
\end{theorem}
However, as Aaronson showed, $\pBQP=\PP$ \cite{Aar05}, 
hence the prover is a $\PP$ machine.  
It is not difficult to apply the fact that $\PsP=\PPP$ to arrive at
the well known result of \cite{LFKN92}\footnote{While the original result might be more closely regarded as $\PsP\in \IP[\BPP, \sP]$, using $\PsP=\PPP$ we get an equivalence, as can be seen in Section \ref{sec:results}}: 
 \begin{theorem}\label{thm:maintheorem2}
$\PsP\in \IP[\BPP, \PP]$. 
Furthermore, the $\IP[\BPP, \PP]$ protocol has completeness 1.
\end{theorem}
The derivation of Theorem \ref{thm:maintheorem2} from
Theorem \ref{thm:maintheorem1} first uses $\PPP=\PsP$ \cite{ZOO}, 
and works for a language in $\PPP$. It then relies on the verifier 
simulating the $\p$ machine 
for any language in $\PPP$ and invoking the interactive protocol 
of Theorem \ref{thm:maintheorem1}
whenever $\p$ would have called the $\PP$ oracle. The detailed proof can be found in appendix \ref{app:T1.2}.

Thus, we arrive at a new proof of the celebrated theorem of \cite{LFKN92}, 
and a new sum-check protocol. The protocol is very natural to use when 
considering quantum related problems, since it applies directly to 
evaluating products of local projections and local unitary 
matrices which appear naturally in quantum complexity. 
Considering the classical setting, our protocol can be used for a $\PsP$ problem, as in \cite{LFKN92}, in the following way: given a $\PsP$ problem, one 
would need to map the problem first into a $\PPP$ problem, 
using $\PsP=\PPP$ \cite{ZOO}, then map each of the $\PP$ problems into a $\pBQP$ circuit, (this
can be done using an easy combinatorial reduction, by \cite{Aar05})
and then invoke the interactive protocol for each of the $\pBQP$ problems.

\subsection{An overview of the protocol} \label{sec:protocol overview}
Our protocol differs from that of \cite{LFKN92, Sha92} in a few crucial points, 
but still exhibits the same overall structure of the original 
sum-check protocol \cite{LFKN92, Sha92}, of inductive consistency checking along 
a path from the root of a tree to a leaf. Let us review this overall 
structure since it is useful to have in mind when explaining the new protocol. 
 
At the root of the tree we have a complicated calculation, 
and the children of the root correspond to slightly simpler calculations; 
as we go down the tree, the calculations become simpler and simpler so that 
at a leaf, we reach a calculation that the verifier can 
compute on his own. The  prover first {\it commits} (namely sends to the 
verifier) the values of the calculations 
at the root, and at all its children. The verifier first performs 
some {\it consistency test} between the values of children and that
 of the root, 
and then picks {\it randomly} one of these children, and
challenges the prover to prove that value (which is again a root of a tree, 
except one level shorter). They now continue starting from the new node 
as a root. As the protocol progresses, the verifier goes down along 
one path of this tree towards a leaf, whose value he can verify on his own.  
The main idea is that if the claimed value at a certain node $q$ was wrong 
then in order for the consistency check between that node and the values 
of its children to pass with high probability, 
many of the children's values must also be 
wrong. So with high probability, if the prover committed to a wrong claim, 
he is still committed to a {\it wrong}
claim also in the next round - so he is ``caught in his lies'' until the 
verifier and prover  
reach a simple claim which the verifier can calculate by
himself, at which point he will catch 
the lie.   

In the sum check protocol, the claimed value was the value of the exponential
sum of the evaluations of a polynomial in $n$ variables $b_1, \dots, b_n$: $\sum_{b_1\in\{0,1\}}\dots\sum_{b_n\in\{0,1\}} P(b_1,\dots,b_n)=K$.

Here, our claimed value at the root would be the trace of what we call a \textit{top row matrix} (see definition \ref{def:toprow}). This is an exponential size matrix 
$A=\ket{0^n}\bra{0^n}\cdot G_T\cdot G_{T-1}\cdots G_1$, where each 
of the $G_i$ is a quantum gate acting non-trivially on at most $3$ 
out of $n$ qubits. That value too can be viewed as an exponential sum. 

\paragraph{Protocol overview:} 
Imagine we are a $\BPP$ verifier, and we would like to verify 
that the value of $\tr{A}=\tr{\ket{0^n}\bra{0^n}\cdot G_T\cdots G_1}=C$, 
by interacting with a $\pBQP$ prover. 
To do this, the prover and verifier will keep updating the matrix $A=A_0$; 
at the $i$th round, the verifier will pick a random unitary $U_i$ acting 
non-trivially only on the three qubits that the gate $G_i$ works on, 
and the prover will be challenged 
to prove the trace of the updated matrix $A_i$, defined to be 
$A_i=U_i\cdot U_{i-1} \cdots U_1 \ket{0^n}\bra{0^n} G_T\cdots G_{i+1}$. 

The interaction goes as follows: 
each round, the verifier asks the prover to 
provide him with the reduced density 
matrix of $A_i$ on the set of qubits which the next gate, $G_{i+1}$,  
works non-trivially on (The first round is the $0$'th round, at which the prover is asked to provide $A_0$ reduced to the qubits of $G_1$). 
His consistency check at round $i>0$ 
is to compare between the {\it trace}
of this reduced matrix, denoted $M_i$, and the trace 
of a three qubit matrix derived by some simple 
manipulation involving a rotation by $U_i$, 
on the matrix $M_{i-1}$ he recieved from the prover in the previous round. 

Notice that it is critical for our protocol, that the 
randomness is provided by matrices $U_i$ which are not general 
$3$-qubits unitaries, but which have a very particular structure:  
they are {\it tensor products of random single 
qubit matrices}. Since the multiplication of such matrices ($U_T\cdots U_1$) 
is still in the form of tensor product of single qubit matrices, 
this is what allows the verifier at the end to compute 
the trace of $A_T=U_T\cdot U_{T-1}\cdots U_1 \ket{0^n}\bra{0^n}$
on his own. 

We must prove that this structure provides 
sufficient randomness to catch a cheating prover; namely, that the prover
could not give us a false original reduced matrix with a false trace, 
and then in the next round, switch to a reduced matrix which is true, but 
will pass the consistency test.  

The soundness relies on the following central basic fact. Consider round 
$i$. Suppose the difference between the correct reduced 
matrix $M_i$ and the one 
the prover actually sent (which we denote by $M'_i$),
is $\Delta_i=M_i-M'_i$.

The fact that the randomness in tensor product unitaries is sufficient to 
detect that this difference is non-zero, is captured by the following 
Lemma \ref{lm:basic}, central to this work  
(and proven in appendix \ref{app:L1}): 
\begin{lemma}\label{lm:basic}
Let $\Delta$ be an operator on $n$ qubits and let $u=u^1\otimes\cdots\otimes u^n$ be a unitary operator on $n$ qubits:
\[
\Delta\neq0 \Rightarrow Pr_{\substack{u^1,\cdots,u^n\backsim U(2)\\
u=u^1\otimes\cdots\otimes u^n}}\Big(tr(\Delta\cdot u)=0\Big)=0,
\]
\end{lemma}
Where $U(2)$ in the above is a probability distribution over 
single qubit matrices, which is closely related to the Haar measure
but not equal to it (see Definition \ref{def:U(2)}). The reason not to work with Haar measure is that later, when handling precision issues - see below - the distribution $U(2)$ seems simpler to work with.
This lemma allows us to prove that if the prover cheated in an 
earlier round, he must continue to cheat in the next round, or the consistecy 
check will detect an ``inconsistency''. 

{~}

\noindent{\bf Handling Precision errors}
Of course, all the above discussion neglected the fact that the 
verifier and prover cannot work with infinite precision, namely, with real 
numbers and coninuous probability measures.  
We modify the protocol by approximating all the above to include exponentially 
small errors (where also the consistency test is performed up to
exponentially small errors); 
thus, the communication between the verifier and prover 
is ploynomially bounded. The proof that the entire analysis goes through 
is much more tedious than the proof for the basic protocol 
which assumes infinite 
precision calculations; however it uses only simple algebraic manipulations. 
The main idea in the proof is given by Claim \ref{cl:soundness by induction}. In this claim, a very important observation is made. It turns out that the prover in the protocol can always reduce the trace of the error matrix $\Delta$ by 
some constant factor each round. Our main lemma for the approximation case, Claim \ref{cl:soundness by induction} states that if the prover is not caught in the protocol, it must mean that the error is not decreasing faster than some exponent in the number of rounds. Since the precision of the calculations of the verifier at the final stage is exponentially good, the error cannot 
decrease below that precision, and thus the cheating will be detected. 

It should be mentioned that this effect in which the prover can reduce its error exponentially in the number of rounds, is very important for the possible application to verification with a $\BQP$ prover. See Subsection \ref{sec: discussion} for further discussion on this point.

\subsection{Comparison to the Sum-Check protocol \cite{LFKN92, Sha92}, 
and other Related work} 
A crucial difference between the known classical protocol \cite{LFKN92, Sha92}
and ours is this.  
In the classical sum-check protocol, the final verification at a leaf is always of a {\it single value} out of exponentially many in the sum. 
In our protocol, on the other hand, 
the computation that the verifier needs to perform 
at the leaf is of the form  
\[tr(U_T\cdots U_2 \cdot U_1\ket{0^n}\bra{0^n})=
tr(w_n\otimes ... \otimes w_2 \otimes w_1 \cdot\ket{0^n}\bra{0^n})
 \]
for single qubit matrices $w_i$. Notice that naively, this value    
is the sum of {\it exponentially many} 
terms. It is only because of the tensor product structure of this 
expression, that this exponential sum can actually be calculated 
efficiently by the verifier. The fact that the verification 
at the leaf is of an exponential sum, and, relatedly,  the tensor product 
structure of the protocol, 
which is naturally quantum, is an inherent difference between 
this protocol and the classical ones. 

Our protocol has similarities in structure with
the interactive proof protocol for $\QMA$ in
\cite{AAV13}.  
The idea in \cite{AAV13} is to consider the local Hamiltonian $H$, 
which is the input to the $\QMA$ problem, and 
the result of the $\QMA$ computation is encoded in the trace of a highly 
non-local matrix $f(H)$ which involves polynomial powers of $H$. 
The randomness in the protocol is provided not by 
 random unitaries but by random 
{\it projections}, each time on a single qubit. 
We note that though implicit in \cite{AAV13}, 
we believe that the essence of the security of the \cite{AAV13} protocol
could also be highlighted by a similar
lemma to Lemma \ref{lm:basic}, except for random projections 
rather than random unitaries; 
in both protocols randomization is done by single qubit 
manipulations. The explanation of the protocol of \cite{AAV13} using 
a similar lemma to Lemma \ref{lm:basic} remains to be done.  

Our work improves on \cite{AAV13} in several aspects.
Most importantly, our result applies not only to $\QMA$ but to
any problem in $\PsP$, thus rederiving the celebrated 
classical result of \cite{LFKN92}, 
and connecting our (quantum-inspired though classical) protocol to the world of classical 
complexity classes. We do not know how to directly extend the protocol of 
\cite{AAV13} to apply for $\PsP$, due to it strongly relying on the 
structure of local Hamiltonians. Another difference is that the protocol presented here itself is considerably simpler than 
that of \cite{AAV13}, when applied to a problem in $\BQP$; Moreover, when applied to a quantum circuit, our protocol preserves its structure. 
In that context, the \cite{AAV13} protocol goes through a series of 
reductions, in which a $\BQP$ 
instance would be first reduced to a local Hamiltonian instance, which will 
in turn be reduced to a calculation of the trace of
a matrix which is not a local Hamiltonian.  
The result is that \cite{AAV13}
involves a more complicated abstract structure, in which essentially 
the tree is replaced (roughly) by a {\it tree of trees}, 
namely, after a sequence 
of rounds, the parties reach a situation in which the problem becomes slightly 
simpler, and then they need to start again, for the slightly simpler 
problem, and they repeat this for polynomially many times. 
 
\subsection{discussion} \label{sec: discussion}
We hope that our protocol may be helpful 
as a starting point towards showing $\BQP \in \IP[\BPP, \BQP]$, namely 
replacing the prover in our protocol by a $\BQP$ prover. 
This remains as a major open problem and was in fact a
major motivation for this paper. We explain below however why 
resolving the problem would require considerably new ideas. 

First, note that a $\BQP$ prover could crudely {\it approximate}
a $\pBQP$ prover by an additive approximation. 
This could naively suggest that when applying the protocol
to a problem in $\BQP$, 
the exponential precision in our protocol is not needed; 
it is sufficient to request that the prover only provides 
approximate answers, and the verifier accepts if the consistency 
checks are approximately satisfied.  
Hence, perhaps, if only additive approximation is required,  
the prover in our protocol can be replaced 
by a $\BQP$ prover? If possible, that would resolve the
above major problem of verifying $\BQP$ computations using a classical 
$\BPP$ verifier.  

Unfortunately, there is a serious problem with this suggestion. 
The problem has to do with exactly the exponential decay in the error 
mentioned earlier. Let us explain the issue very roughly below 
(compromising rigor to make the argument clearer). Recall that 
in our protocol, consistency checking is done by comparing only  
the {\it traces} of two matrices (at each round).  
Let us assume that the prover wants to prove a wrong claim, 
namely that the trace of the exponential matrix $A$ 
is in fact $C$ where $C$ is different from the correct value 
$tr(A)$ by $\delta_0$; 
In other words, the prover sends in the first round $M'_0$ such that 
$tr(M'_{0})=C$; recalling that the correct matrix is denoted $M_0$, we have 
$tr(M'_0)=tr(M_0)+\delta_0$. We refer to $\delta_0$
as the {\it initial error}. 
The idea is that the prover can always
decrease the error by a constant factor each round, 
and still pass the consistency tests. 
The way the prover will do this, is 
by ``spreading his error'' at each round as much as possible, as follows: 
at each round he will send $M'_{i}=M_i+\Delta_i$ such that 
$\Delta_i=\delta_i I/8$, where $\delta_i$ is his current error.
In other words, he spreads his error evenly on all diagonal elements in 
his matrix. The claim is that because 
of the randomization process, on average, his updated error 
$\delta_{i+1}=\tr{\Delta_i \cdot u_i}$ for the random unitary $u_i$ 
(chosen by the verifier in the following round) is smaller than $\delta_i$ 
by a constant factor. Hence, as we had mentioned earlier in Subsection \ref{sec:protocol overview} - this strategy will make the initial error decrease 
by an exponential factor in the number of rounds (making our analysis 
of the protocol in which this exponential decay of error appears, namely
Claim \ref{sec:soundness by induction proof}, essentially tight). 
Hence, the prover can get away with an 
exponentially large initial error; this is because the consistency tests of 
the verifier at the final round must allow {\it some} inaccuracies if the 
prover is only $\BQP$. Making the protocol sufficiently robust to such 
errors, would lead to a solution to the problem  
but this seems to require significantly 
new ideas. 
We note that the same exponential decay of error is possible to achieve
also in the original $\IP=\PSPACE$ protocol; 
it is a very interesting problem to try and formalize this property of 
interactive protocols more rigorously. 

Given the historical connection between interactive proofs and PCP, 
it is our hope that our protocol will be useful also for progress on the 
quantum PCP front. Another interesting open problem is to extend 
our protocol to all of $\PSPACE$, and derive an alternative proof 
for the full $\IP=\PSPACE$ result; we were not able to do this so far.
It is also our hope that this protocol will prove useful in classical contexts, 
due to its genuinely different structure than previous protocols for the 
same important task.   
 
\subsection{Paper Organization}
We begin with some notations and background 
in Sec. \ref{sec:bckgrnd}. In Sec. \ref{sec:results} we 
formally state our main results and show how they are derived given an interactive protocol $\widehat{W}$ for proving the trace of 
a \textit{top row matrix}. 
In Sec. \ref{sec:W},\ref{sec:infinite} we present our idealized 
interactive protocol $W$ (which assumes infinite precision) and prove its soundness and completeness. In Sec. \ref{sec:bounded} we provide the modified version 
of the protocol, $\widehat{W}$ and analyze it to show that even when
precision limitations are addressed, completeness and soundness hold. 
\section{Background} \label{sec:bckgrnd} 

\subsection{Notations}
We use the following notations throughout the paper:

\begin{itemize}

	\item Given a gate $g_i$ on a set $s_i$ of less than $n$ qubits, we define $G_i=g_i\otimes I$ to be it's $n$ qubits expansion (where $I$ is over the qubits $\overline{s_i}$).
		
	\item Similarly, given a general unitary $u_i$ or $h_i$ on less than $n$ qubits, we define $U_i=u_i\otimes I$ and $H_i=h_i\otimes I$ to be their $n$ qubits expansions.

	\item Given an operator $\rho$ on hilbert space $\ham = \ham_A \otimes \ham_B$, and given an orthonormal basis $\{\ket{i}\}_i$ for $\ham_A$ and 
$\{\ket{l}\}_l$ for $\ham_B$, we can write 
$\rho =\sum_{i,k,j,l}\rho_{i,j,k,l}\ket{i}\bra{j}\otimes \ket{k}\bra{l}$.  
\textit{Tracing out} $\ham_B$, or equivalently reducing $\rho$ 
to $\ham_A$ is defined as: 
\[\rho|_{\ham_A}=tr_{\ham_B}(\rho)=
tr_{\ham_B}\left(
\sum_{i,j,k,l}{\rho_{i,k,j,l}\ket{i}\bra{j}\otimes\ket{k}\bra{l}}\right)=\sum_{i,j}
{\sum_k\rho_{i,j,k,k}}\ket{i}\bra{j}\]
	\item Given a matrix $B_{2^n\times 2^n}$, and a matrix $z$ on a set $s$ of qubits, we denote $B$'s reduction to the qubits $s$ (on which $z$ operates) by $B|_z$. Namely: $B|_z=tr_{\overline{s}}(B)$.
			
	\item $U(2)$ - The measure on single qubit unitaries,
from which we randomly toss our matrices. See Section \ref{sec:W}.  
	\item $\p$ - The complexity class $Polynomial-Time$\footnote{While we use $\p$ for the complexity class, it should not be confused with $P$, which we use to call a Prover throught the paper}. 
	\item In definition \ref{def:truncated unitary} we use the notion $\floor{f(\theta, \phi_1, \phi_2)}_{\xi}$. By this we mean any deterministic approximation of $f(\theta, \phi_1, \phi_2)$, which can be computed by a $\p$ verifier using $\poly{n}$ bits, s.t $\abs{\floor{f(\theta, \phi_1, \phi_2)}_{\xi}-f(\theta, \phi_1, \phi_2)}\leq\xi$
\end{itemize}

\subsection{Top Row Matrix}
A \textit{top row matrix} is defined as follows:
\begin{defn} \label{def:toprow}
Given a $2^n\times2^n$ matrix $B$, we define it's \textit{top row matrix} to be: $\vert0^n\rangle \langle 0^n \vert B$, where $\vert0^n\rangle$ is the computational basis state on $n$ qubits $\vert00\cdots0\rangle$. By referring to the \textit{top row matrix} of a quantum circuit on $n$ qubits which is given as a sequence of local gates 
$g_T\ldots g_1$ we mean the \textit{top row matrix} of 
$B=G_T\cdot G_{T-1}\cdots G_1$.
\end{defn}

\subsection{IP[$\mathcal{V}$, $\mathcal{P}$]}
The complexity class $\IP[\mathcal{V}, \mathcal{P}]$ is defined as such:
\begin{defn} \label{def:ipvp}
Given complexity classes $\mathcal{V}$ and $\mathcal{P}$, $\IP[\mathcal{V}, \mathcal{P}]$ is the class of languages $\lang \subseteq \{0, 1\}^*$ for which there exists a 2 party protocol $\mathcal{W}$ between a verifier of computational power $\mathcal{V}$ and a prover, such that for all inputs $x$,
\begin{enumerate}
\item{Completeness: }If $x \in \lang$, then there exists a prover of computational power $\mathcal{P}$ for which\\ $\prob{verifier-accepts}\geq\frac{2}{3}$.
\item{Soundness: }If $x \notin \lang$, then for any prover, $\prob{verifier-accepts}\leq\frac{1}{3}$.
\end{enumerate}
\end{defn}
\subsection{BQP and Q-CIRCUIT}

\begin{defn}[due to \cite{BV93}] The complexity class $\BQP$ is the set of languages which are accepted with probability 2/3 by some polynomial time Quantum Turing Machine.
\end{defn}

\begin{defn} 
The promise problem Q-Circuit is defined as follows. 
The input is a description of a sequence of gates $g_L,... g_1$, taken out of a finite universal set of local quantum gates, acting non-trivially on at most 3 out of $n$ qubits.  $L$ is polynomial in $n$.  Yes instances and No instances are defined as follows:
\[ Q-CIRCUIT_{YES}: \|(\vert0\rangle \langle0\vert\otimes I_{n-1})\cdot G_L\cdot G_{L-1}\cdots G_1 \vert0^n\rangle \|^2\geqslant\frac{2}{3}
\]
\[ Q-CIRCUIT_{NO}: \|(\vert0\rangle \langle0\vert\otimes I_{n-1})\cdot G_L\cdot G_{L-1}\cdots G_1 \vert0^n\rangle \|^2\leqslant\frac{1}{3}
\]
\end{defn}
Q-CIRCUIT is known to be complete for $\BQP$. We can use for example the universal set of gates consisting only of Toffoli and Hadamard gates \cite{Shi02}.
\begin{claim}\label{cl:top row}
Calculating the trace of the \textbf{top row matrix} 
for an arbitrary sequence of local
gates $g_T\ldots g_1$, acting non-trivially on $3$ out of $n$ input qubits 
such that $T=poly(n)$, to within $\pm\frac{1}{6}$ is hard for $\BQP$.
\end{claim}
We include this simple proof here since its idea will be used several times later on. 
\begin{proof}
Given a Q-Circuit instance, $\prob{0}$ can be written as follows: 
\begin{align*}
&\|(\vert0\rangle \langle0\vert\otimes I_{n-1})\cdot G_L\cdot G_{L-1}\cdots G_1 \vert0^n\rangle \|^2\\ 
&=\langle0^n\vert G^{\dagger}_L\cdot G^{\dagger}_{L-1}\cdots G^{\dagger}_1\cdot (\vert0\rangle \langle0\vert\otimes I_{n-1})\cdot (\vert0\rangle \langle0\vert\otimes I_{n-1})\cdot G_L\cdot G_{L-1}\cdots G_1 \vert0^n\rangle \\
&=\langle0^n\vert G^{\dagger}_L\cdot G^{\dagger}_{L-1}\cdots G^{\dagger}_1\cdot (\vert0\rangle \langle0\vert\otimes I_{n-1})\cdot G_L\cdot G_{L-1}\cdots G_1 \vert0^n\rangle \\ 
&=\langle0\vert \otimes \langle0^n\vert G^{\dagger}_L\cdot G^{\dagger}_{L-1}\cdots G^{\dagger}_1\cdot CNOT_{1\mapsto n+1}\cdot G_L\cdot G_{L-1}\cdots G_1 \vert0^n\rangle\otimes \vert0\rangle \\
&=\tr{\langle0\vert \otimes \langle0^n\vert G^{\dagger}_L\cdot G^{\dagger}_{L-1}\cdots G^{\dagger}_1\cdot CNOT_{1\mapsto n+1}\cdot G_L\cdot G_{L-1}\cdots G_1 \vert0^n\rangle\otimes \vert0\rangle} \\
&=\tr{\vert0^{n+1}\rangle\langle0^{n+1}\vert\cdot G^{\dagger}_L\cdot G^{\dagger}_{L-1}\cdots G^{\dagger}_1\cdot CNOT_{1\mapsto n+1}\cdot G_L\cdot G_{L-1}\cdots G_1 }
\end{align*}

Q-Circuit thus reduces to calculating (to within accuracy $\frac{1}{6})$ the trace of the following \textit{top row matrix}, defined for a sequence of gates consisting of $T=2L+1=\poly{n}$ unitary local gates:
\[ A_{2^{n+1}\times2^{n+1}}= \vert0^{n+1}\rangle\langle0^{n+1}\vert\cdot G^{\dagger}_L\cdot G^{\dagger}_{L-1}\cdots G^{\dagger}_1\cdot CNOT_{1\mapsto n+1}\cdot G_L\cdot G_{L-1}\cdots G_1
\] 
\end{proof}

In Section \ref{sec:bounded} 
we present our interactive protocol $\widehat{W}$, which verifies any such trace to within inverse-exponential accuracy, thus showing an explicit IP for $\BQP$. Furthermore, we show that the protocol can be carried out by a $\pBQP$ prover.

\subsection{PostBQP} \label{sec:PostBQP}
Introduced by Aaronson in \cite{Aar05}, the complexity class $\pBQP$ consists of all of the computational problems solvable in polynomial time on a quantum Turing machine with postselection and bounded error. Formally, the class can be defined as such:
\begin{defn}\label{def:postbqp}$\pBQP$ is the class of languages $\lang \subseteq \{0, 1\}^*$ for which there exists a uniform family of polynomial-size quantum circuits $\{C_n\}_{n\geq 1}$
such that for all inputs $x$,
\begin{enumerate}
\item{}After $C_n$ is applied to the state $\ket{0\cdots 0}\otimes\ket{x}$, the first qubit has probability $\geq \frac{1}{2^n}$ of being measured $\ket{0}$
\item{}If $x \in \lang$, then conditioned on the first qubit being $\ket{0}$, the second qubit is $\ket{0}$ with probability at least 2/3.
\item{}If $x \notin \lang$, then conditioned on the first qubit being $\ket{0}$, the second qubit is $\ket{0}$ with probability at most 1/3.
\end{enumerate}
\end{defn}
We remark that this definition is slightly different than the 
the original definition of \cite{Aar05}, where the probability in item $1$ was only required to be non-negative. As noted by Aaronson later in \cite{AarPost}
this corrected definition is in fact the one to be used, as only for this 
definition do we know that 
$\pBQP$ is equal to $\PP$. Moreover, this equality holds when the quantum circuit is assumed to consist of only Hadamard and Toffoli gates.

We also remark the simple fact that $\pBQP$ is closed under complement. 
This can easily be seen by flipping the second qubit prior to measuring.

A claim that will be important for us in order to show an interactive protocol for all languages $\lang$ in $\pBQP$ is the following.
\begin{claim} \label{cl:postBQPtop}
Any language $\lang \in \pBQP$ can be decided by calculating a division $y/z$ for two real values 
$y$ and $z$ which are both 
 traces of \textbf{top row matrices}, each  
for an arbitrary sequence of local unitary quantum gates $g_T\ldots g_1$, 
acting non-trivially on (at most)
three out of $n$ input qubits such that $T=poly(n)$,
and are each calculated to within accuracy $\epsilon=\frac{1}{10\cdot2^n}$.
\end{claim}

\begin{proof}
Given a language $\lang\in\pBQP$ with input $x$ of $n$ bits, consider the 
quantum circuit $C_n$ as in Definition \ref{def:postbqp}, 
where the number of its ancilla qubits 
is $n_{Anc}$ and it acts on $n'=n+n_{Anc}$ qubits.
Consider the final state of the circuit, $C_n\left(\ket{00\cdots 0}\otimes\ket{x}\right)$. Denote by $a$ the outcome of a measurement of the second qubit in that state in the computational basis, and by $b$ the measurement outcome when measuring the first qubit in that state in the computational basis.

\begin{itemize}
\item{if $x\in\lang$:} $\prob{a=\ket{0}|b=\ket{0}}\geq\frac{2}{3}$
\item{if $x\notin\lang$:} $\prob{a=\ket{0}|b=\ket{0}}\leq\frac{1}{3}$
\end{itemize}

Consider the sequence of quantum gates $g_1,....,g_L$ describing the  
simple quantum computation consiting of first mapping
$\ket{0^n}\mapsto \ket{x}$ by flipping the required bits one by one, and then applying the circuit $C_n$. 

It holds that:

\[
\prob{a=0|b=0}
=\frac{\prob{a=0\cap b=0}}{\prob{b=0}}
=\frac{\|(\ket{00}\bra{00}\otimes I_{n-2})\cdot G_L\cdot G_{L-1}\cdots G_1 \vert0^{n'}\rangle \|^2}
	  {\|(\vert0\rangle \langle0\vert\otimes I_{n-1})\cdot G_L\cdot G_{L-1}\cdots G_1 \vert0^{n'}\rangle \|^2}
\]

And a similar argument to that used in proof of Claim \ref{cl:top row} shows that both the nominator and denominator can be calculated by calculating the trace of a 
\textit{top row matrix} which comprises $\poly{n}$ local gates.
 We denote the value in the nominator $y$, and the value in the denominator $z$, we keep in mind that $z\geq 10\epsilon$ and go back to look at the two case:

\begin{itemize}
\item{if $x\in\lang$:} $\frac{y\pm\epsilon}{z\pm\epsilon}
			\geq\frac{y-\epsilon}{z+\epsilon}
			\geq\frac{y-\frac{1}{10}z}{z+\frac{1}{10}z}
			=\frac{10}{11}(\frac{y}{z}-\frac{1}{10})
			\geq\frac{10}{11}(\frac{2}{3}-\frac{1}{10})>\frac{1}{2}+0.01$
\item{if $x\notin\lang$:} $\frac{y\pm\epsilon}{z\pm\epsilon}
			\leq\frac{y+\epsilon}{z-\epsilon}
			\leq\frac{y+\frac{1}{10}z}{z-\frac{1}{10}z}
			=\frac{10}{9}(\frac{y}{z}+\frac{1}{10})
			\leq\frac{10}{9}(\frac{1}{3}+\frac{1}{10})<\frac{1}{2}-0.01$
\end{itemize}

And so, by calculating the two \textit{top row matrices} $y$ and $z$
with accuracy $\epsilon$, we can differentiate between the cases and the 
claim holds. 
\end{proof}

\section{Derivation of Main Results} \label{sec:results}
We now show how our main results, Theorems \ref{thm:maintheorem1} and \ref{thm:maintheorem2}, can be derived from an interactive proof for the value of a \textit{top row matrix}.
In Section \ref{sec:bounded} we provide an interactive protocol $\widehat{W}$, with a BPP verifier, in which an honest $\pBQP$ prover can
prove the correctness of a value $C$ claimed to be  
the trace of a \textit{top row matrix}, 
to within inverse exponential accuracy. 
We will prove in that section (section \ref{sec:bounded}) that
$\widehat{W}$ has completeness $1$ even if the prover $P$ is 
computationally limited to $\pBQP$. 
We prove that $\widehat{W}$ has soundness at most $\frac{1}{3}$ against any prover in the same section. 

We remark here that throughout the paper, in any place where the ($\pBQP$) prover $P$ is said to compute entries or traces of matrices to within exponential precision,
what we mean is that those values are recovered by the verifier $V$ by a sequence of (polynomially many) binary search queries, as the prover's computational power is assumed to be the decision class $\pBQP$.

\paragraph{Theorem \ref{thm:maintheorem1}}
$\pBQP\in \IP[\BPP, \pBQP]$\\
\textit{Furthermore, the $\IP[\BPP, \PP]$ protocol above has completeness 1.}
\begin{proof}
By Claim \ref{cl:postBQPtop} it suffices to verify the traces of two
\textit{top row matrices} to within $\pm\frac{1}{10\cdot 2^n}$ 
in order to decide on any instance of a language $\lang \in \pBQP$. 
By Claims \ref{cl:bounded completeness} and \ref{cl:bounded soundness} the protocol $\widehat{W}$ does just that for a single \textit{top row matrix}, with completeness 1 (even when the prover is restricted to be a $\pBQP$ machine) and soundness at most $\frac{1}{3}$. By repeating the protocol $\widehat{W}$ twice and accepting if and only if both runs of $\widehat{W}$ pass (for each of the two \textit{top row matrices}), the soundness of each of the \textit{top row matrices} is reduced to $\frac{1}{9}$, while maintaining completeness 1. By the union bound, this means that conditioned on both (repeated) protocols passing, the probability for one the values of the traces of not be within $\pm\frac{1}{10\cdot 2^n}$ of its actual value is at most $\frac{2}{9}<\frac{1}{3}$. This proves Theorem \ref{thm:maintheorem1}.
\end{proof}

\paragraph{Theorem \ref{thm:maintheorem2}}
$\PsP\in \IP[\BPP, \PP]$. 
\textit{Furthermore, the $\IP[\BPP, \PP]$ protocol above has completeness 1.}\\
{\bf Proof Sketch:}
We first note that the protocol of Theorem \ref{thm:maintheorem1} can be repeated so that the error (soundness) is exponentially reduced. Now, for any language in $\PsP=\PPP=\PpBQP$ we construct a protocol where each query to the $\pBQP$ oracle is simulated by using the repeated protocol of Theorem \ref{thm:maintheorem1}. The repeated Theorem \ref{thm:maintheorem1} protocol is used both for the original query and for its complement. The original query is used for verifying that the answer is 1 (if that is the case), and if it is not the case - the complement query is used to verify that the answer is 0. The combination gives the correct (oracular) answer with probability exponentially close to 1 (in the number of repetitions), so by using the union bound - we get that with high probability all the oracle call simulations give the correct results. This in turn means that the verifier, which simulates the $\PpBQP$ algorithm, will deduce the correct answer. The complete proof of Theorem \ref{thm:maintheorem2} can be found in Appendix \ref{app:T1.2}.

\section{The Interactive protocol (assuming unlimited precision)}
\label{sec:W}
We will now present our first protocol, which verifies the trace of a \textit{top row matrix} $A$ for a quantum computation given by a sequence of local 
quantum gates $g_1,\cdots ,g_T$.

That is, our protocol verifies 
$\tr{A}=C$ 
for $A=\vert0^n\rangle \langle 0^n \vert G_T\cdot G_{T-1}\cdots G_1$ and value $C$. This first protocol, denoted $W$, assumes infinite precision, and works with
completeness 1 and soundness 0. 

The protocol begins by the verifier asking the prover to 
send the three qubit matrix $M_0$, which is the matrix $A=A_0$ reduced to the qubits of the first 
gate $g_1$. Of course, the prover may not send the correct matrix; 
Denote by $M'_0$ the actual matrix the prover sends. 
The verifier then performs his first consistency check: 
he checks that the {\it trace $M'_0$}
is indeed $C$, as expected. 

Next, the verifier chooses a random
unitary $U_1$ acting on the three qubits of the gate, $G_{1}$, 
which will be used soon to ``replace'' the gate $G_1$ in the matrix. 
Importantly, 
$U_1$ is chosen to be a tensor product of random single qubit matrices. 
$U_1$ will be used to define a new target computation $A_1$, 
$A_1$ is essentially defined to ``peel off'' $G_1$ from $A=A_0$ 
and replace it by the random $U_1$, in a cyclic manner: 
$A_1=U_1 \cdot \vert0^n\rangle \langle 0^n \vert G_T\cdot G_{T-1}\cdots G_2$. 

Now, the prover is supposed to send $M_1$, which is 
$A_1$ reduced to the three qubits of the {\it next} gate, $G_2$; 
The reduced matrices that the verifier has, $M'_0$ and $M'_1$, 
don't even act on the same set of qubits, but if correct, 
their {\it traces} can be related, which is the consistency check that the 
verifier performs. 

If they pass, the verifier and prover pass to ``peeling off'' the next gate: 
the verifier now picks a randomization matrix $U_2$ for the qubits 
of $G_2$, asks for another matrix $M_2$, and so on. 

At the end of this process, because 
all $U_i$ are tensor product matrices, 
the verifier is left with a calculation composed strictly of tensor product operators, which he can compute by himself.

The soundness of the protocol relies on the fact 
that if $M'_i$ is different than $M_i$
the prover must be extremely lucky (in the choice of $U_i$) to be able to send
$M'_{i+1}= M_{i+1}$ 
and still pass the consistency check.  
Thus he is caught in his lies, and must continue cheating also in the 
next round (and send $M'_{i+1}\neq M_{i+1}$). If he keeps doing this, 
he will get caught in the final round which is a computation the verifier 
can perform.\\

Before moving to a the full description of the protocol, we will define $U(2)$ - the measure on single qubits unitaries from which unitaries in the protocol are drawn.

\begin{defn} We define the following measure \label{def:U(2)}
on unitary operators on a single qubit, denoted $U(2)$: 
choose angles $\theta, \varphi_1, \varphi_2$ uniformly at
random from $[0,2\pi]$, and define the resulting unitary as:

\[
	\begin{pmatrix}   
	\cos\theta\cdot e^{i\varphi_1} & \sin\theta\cdot e^{i\varphi_2} \\
	-\sin\theta\cdot e^{-i\varphi_2} & \cos\theta\cdot e^{-i\varphi_1}
	\end{pmatrix}
\]
\end{defn}
We note that this measure is reminiscent of the Haar measure on single qubit unitaries, Except that according to the Haar measure the distribution of the angle $\theta$ is different \cite{Ozo09}. We use this  
modified measure as it simplifies the calculations.

We now give a detailed description of the protocol $W$, where the verifier $V$ is given a quantum \textit{top row matrix} $A_{2^n\times 2^n}=\vert0^n\rangle \langle 0^n \vert\cdot G_T\cdot G_{T-1}\cdots G_1$, and interacts with a prover $P$ (which we assume here in this first protocol to be unbounded) to verify $\tr{A}=C$:

\begin{framed}
	\paragraph{The protocol W:}
	\begin{enumerate}
		\item {\em In the 0'th round} -- $V$ asks for $M_0=A|_{g_1}$, receives back a matrix $M'_0$, and verifies that $C=\tr{M'_0}$ (rejects otherwise).
 		\item {\em In the $i$'th round} -- $V$ chooses $u_i^1, u_i^2, u_i^3 \backsim U(2)$, sets $u_i = u_i^1\otimes u_i^2 \otimes u_i^3$ on the qubits on which $g_i$ operates, and asks for:
		 \[
		 M_i=(U_i\cdot U_{i-1}\cdots U_1 \cdot \vert0^n\rangle \langle 0^n \vert\cdot G_T\cdots G_{i+1})|_{g_{i+1}}
		\]
		$V$ receives back a matrix $M'_i$, and verifies that $\tr{M'_i}=\tr{M'_{i-1}\cdot g_i^{-1}\cdot u_i}$.
		\item {\em In the T'th round} -- $V$ chooses $u_T$ as before, and accpets if $\tr{U_T	\cdot U_{T-1}\cdots U_1 \cdot \vert0^n\rangle \langle 0^n \vert}=\tr{M_{T-1}\cdot g_T^{-1}\cdot u_T}$. (rejects otherwise).
	\end{enumerate}
\end{framed}
\section{Completeness and Soundness for W} \label{sec:infinite} 
\begin{theorem}\label{thm:baby}
Assuming $P$ and $V$ can represent, communicate, and perform calculations on values with infinite precision, the protocol $W$ has completeness 1 and soundness 0.
\end{theorem}

The main idea in our protocol is the fact that the $u_i$ unitaries are simple enough to allow $V$ to accurately compute $\tr{U_T\cdot U_{T-1}\cdots U_1\cdot \vert0^n\rangle \langle 0^n \vert}$ - this is true as $\vert0^n\rangle \langle 0^n \vert=\vert 0\rangle \langle 0\vert \otimes \vert 0\rangle \langle 0\vert \otimes \cdots \otimes \vert 0\rangle \langle 0\vert$, and so the matrix $U_T \cdot U_{T-1} \cdots U_1 \ket{0^n}\bra{0^n}$ is in fact a product of matrices, where each of these matrices is a tensor product of single qubit matrices. But still, the $u_i$ unitaries provide enough randomness for our consistency checking to be effective. That is, if P did not provide us the $M_{i-1}$ matrix we asked for at round $i-1$, he will not pass the consistency checking $\tr{M_i}=\tr{M_{i-1}\cdot h_i}$ with the $M_i$ we asked for. 
The proof of completeness is given in subsection \ref{sec:w completeness}, and that of soundness in subsection \ref{sec: w soundness}.
\subsection{Completeness}
\label{sec:w completeness}
	
	To prove the protocol's completeness, we will first prove the following basic Lemma \ref{lm:reduced assosiativeness} (proven in appendix \ref{app:L3}):

	\begin{lemma}\label{lm:reduced assosiativeness}
	Let $U$ an operator on a hilbert space \(\mathcal{H}\), and $q$ an operator on a subsystem $Q$ of \(\mathcal{H}\):
	\[
	U|_Q\cdot q=(U\cdot(q\otimes I_{\overline{Q}}))|_Q
	\]
	\end{lemma}

	\begin{defn} We define $P$'s \textit{Truthful Strategy} to be:
		\begin{enumerate}

			\item {\em In the 0'th round} -- $P$ sends back $A|_{g_1}$.
  
	 		\item {\em In the $i$'th round} -- $P$ receives $u_i$ from $V$, and sends back $(U_i\cdot U_{i-1}\cdots U_1 \cdot \vert0^{n}\rangle\langle0^{n}\vert\cdot G_T\cdots G_{i+1})|_{g_{i+1}}$.

		\end{enumerate}
	\end{defn}

	\begin{claim}\label{cl:w completeness}
	if $\tr{A}=C$, $P$ passes $W$ with probability 1
	\end{claim}
	\begin{proof} We show that using it's \textit{Truthful Strategy}, $P$ passes each of $V$'s verification stages:

		\begin{subclaim}\label{cl:w round 0}		
		$V$'s verification passes the $0$'th round.
		\end{subclaim}
		\begin{proof} This is obvious, as 
			\[
			C=\tr{A}=\tr{A|_{g_1}}
			\] 

		\end{proof}

		\begin{subclaim}\label{cl:w round i}$V$'s verification passes the $i\in[T-1]$'th round.
		\end{subclaim}
		\begin{proof} This is true, as using Lemma \ref{lm:reduced assosiativeness}:
		
		\begin{align*}
			\tr{M_{i-1}\cdot g_i^{-1}\cdot u_i}
			&=\tr{(U_{i-1}\cdot U_{i-2}\cdots U_1 \cdot \vert0^{n}\rangle\langle0^{n}\vert\cdot \cdot G_T\cdots G_{i})|_{g_{i}}\cdot g_i^{-1} \cdot u_i}\\
			&=\tr{U_{i-1}\cdot U_{i-2}\cdots U_1 \cdot \vert0^{n}\rangle\langle0^{n}\vert \cdot G_T\cdots G_{i}\cdot G_i^{-1}\cdot U_i}\\
			&=\tr{U_{i-1}\cdot U_{i-2}\cdots U_1 \cdot \vert0^{n}\rangle\langle0^{n}\vert \cdot G_T\cdots G_{i-1}\cdot U_i}\\
			&=\tr{U_i\cdot U_{i-1}\cdot U_{i-2}\cdots U_1 \cdot \vert0^{n}\rangle\langle0^{n}\vert \cdot G_T\cdots G_{i-1} }\\
			&=\tr{M_i} \numberthis \label{tr(M g u)}
		\end{align*}

		\end{proof}	

		\begin{subclaim}\label{cl:w round T}$V$'s verification passes the $T$'th round.	
		\end{subclaim}
		\begin{proof}
		\begin{align*}
			\tr{M_{T-1}\cdot g_T^{-1}\cdot u_T}
			&=\tr{(U_{T-1}\cdot U_{T-2}\cdots U_1 \cdot \vert0^{n}\rangle\langle0^{n}\vert \cdot G_T)|_{g_{i}}\cdot g_T^{-1}\cdot u_T}\\
			&=\tr{U_T\cdot U_{T-1}\cdots U_1 \cdot \vert0^{n}\rangle\langle0^{n}\vert}
		\end{align*}
		\end{proof}
	
		And so the proof for Claim \ref{cl:w completeness} is established.

	\end{proof}

\subsection{Soundness} \label{sec: w soundness}
	\begin{claim}\label{cl:w soundess}
	if $\tr{A}\neq C$, $P$ passes $W$ with probability 0.
	\end{claim}
	\begin{proof}
		First, let us define for each round $i$ 
the three qubits matrix $\Delta_i$ which is the error matrix for that round. $\Delta_i$ is 
the difference between the correct matrix $M_i$ and
the matrix $M'_i$ that the prover actually sent: $\Delta_i=M_i-M'_i$.

Using this notation, and assuming $P$ passes the $0'th$ round, the assumption $\tr{A}\neq C$ translates into:
		\begin{equation}\label{eq:tr(delta0)=0}
		\tr{\Delta_0}=\tr{M_0-M'_0}=\tr{M_0}-\tr{M'_0}=\tr{A}-C\neq 0 
		\end{equation}
		
		So $\Delta_0\neq 0$. In order for $P$ to pass the tests in 
$W$, $P$ must pass all of the rounds $\{1,\cdots,T-1\}$, and either:
		\begin{enumerate}
			\item pass round $i$
 (for some $0<i<T$) with $\Delta_{i-1}\neq 0$ and $\Delta_{i} =0$; or: \label{enum1}
			\item pass the $T'th$ round with $\Delta_{T-1}\neq0$ \label{enum2}
		\end{enumerate}

Consider the first case: for some $0<i<T$ 
we have $\Delta_{i-1}\neq 0$ while $\Delta_{i}=0$. 
We claim that the probability (over the choice of $u_i$)
for $P$ to pass the consistency test 
$tr(M'_{i-1}\cdot g_i^{-1}\cdot u_i)=tr(M'_i)$ is $0$. 
This is because by the linearity of the trace, 
and using the fact that $tr(M_{i-1}\cdot g_i^{-1}\cdot u_i)=tr(M_i)$ 
(as we have seen in the completeness analysis)  
this entails:
		\begin{equation}\label{eq:tr(delta*h)=tr(delta)}
		tr(\Delta_{i-1}\cdot g_i^{-1}\cdot u_i)=tr(\Delta_i)=0 
\end{equation} 
	As $g_i^{-1}$ is unitary, we have 
$\Delta_{i-1} \neq 0 \Rightarrow \Delta'=\Delta_{i-1}\cdot g_i^{-1}\neq 0$, so we can use Lemma \ref{lm:basic} with the non-zero matrix 
$\Delta'$, to show that the probability for this over the choice of 
$u_i$ is $0$. 

Hence, we can deduce that if $P$ passed the tests at all rounds up to 
round $T$, it must be that $\Delta_{T-1}$ is nonzero. 
We now remember that $tr\left(M_{T-1}\cdot g_T^{-1}\cdot u_T\right)
			=tr\left(U_T\cdot U_{T-1}\cdots U_1 \cdot \vert0^{n}\rangle\langle0^{n}\vert\right)$, and the prover needs to pass the test
$tr\left(M'_{T-1}\cdot g_T^{-1}\cdot u_T\right)
		=tr\left(U_T\cdot U_{T-1}\cdots U_1 \cdot \vert0^{n}\rangle\langle0^{n}\vert\right)$. 
For this we must have 
$tr\left(\Delta_{T-1}\cdot g_T^{-1}\cdot u_T\right)=0$. Given that $\Delta_{T-1}\cdot g_T^{-1}$ is a non zero matrix, we have 
again that the probability for $u_T$ to satisfy this is zero, 
using Lemma \ref{lm:basic}. 
This means that if $P$ passes the first $T-1$ rounds, 
he has zero probability to pass the final round. 
This concludes the proof of Claim \ref{cl:w soundess}.

	\end{proof}	
	Having shown both the completeness and soundness parameters for $W$, we thus conclude the proof of Theorem \ref{thm:baby}. 

\section{Handling precision errors}
\label{sec:bounded}
The previous discussion of the protocol $W$ made the implicit assumption that all values can be computed and communicated over a classical channel with infinite accuracy. But of course, this assumption is false. The prover and verifier 
cannot communicate real numbers with infinite precision, 
the verifier cannot truly sample a unitary from the modified 
Haar measure with infinite precision.
Indeed, the protocol $W$ was introduced for didactic purposes; We now show 
how it can be slightly adjusted such that it uses only polynomially
many bits of classical communication, and can be carried out efficiently. 
We denote the adjusted protocol $\widehat{W}$. 

The modifications in the protocol are of several types, and we will 
use several {\it accuracy parameters} to describe them. 
First, since everything is going to work up to 
some accuracy, the consistency checking can no longer be done
using equality tests. Instead, each test for equality in $W$'s verification, 
is checked up to a small accuracy parameter $\mu=\frac{K}{4\chi^T}$ 
where $K=\frac{1}{10\cdot 2^n}$ and $\chi=60^{12}T^9$ (where $T$ is the number of gates). 
Another change is that the $u_i^j$ unitaries cannot be
chosen from the measure $U(2)$, 
as this requires infinite precision. Instead, we introduce another 
accuracy parameter $\xi=\frac{\mu}{2^{2n+11}T}=\Theta(\frac{1}{2^{n^{c}}})$, 
and the parameters defining the unitary matrix entries are 
defined up to this $\xi$: 
\begin{defn} We define choosing a unitary operator on a single qubit according to the \textit{truncated} measure as 
choosing angles $\theta, \varphi_1, \varphi_2$ with precision $\xi$ uniformly at random from $[0,2\pi]$ (By this, we mean that $\theta, \varphi_1, \varphi_2$ are chosen uniformly at random from $\{i\cdot \xi \ \big\vert \  i\in\mathbb{N},\ i\cdot \xi<2\pi\}$). These angles are the representation for the unitary

\[
	\begin{pmatrix}   
	\cos\theta\cdot e^{i\varphi_1} & \sin\theta\cdot e^{i\varphi_2} \\
	-\sin\theta\cdot e^{-i\varphi_2} & \cos\theta\cdot e^{-i\varphi_1}
	\end{pmatrix}
\]

We denote generating an operator $u$ on a single qubit according to this altered measure by: $u\backsim \widehat{U}(2)$
\end{defn}
One final issue we have to address, is the fact that given a unitary $u$ on a single qubit, we can not assume $V$ nor $P$ can perform $u$ accurately. We introduce the following approximation notion:
\begin{defn}
\label{def:truncated unitary}
Given a unitary operator 
\[ u=
	\begin{pmatrix}   
	\cos\theta\cdot e^{i\varphi_1} & \sin\theta\cdot e^{i\varphi_2} \\
	-\sin\theta\cdot e^{-i\varphi_2} & \cos\theta\cdot e^{-i\varphi_1}
	\end{pmatrix}
\] on a single qubit, we define the \textit{truncated unitary} $\widehat{u}$ to be a version of $u$ where each each of the entries is calculated up to an accuracy of $\xi$:
\[ \widehat{u}=
	\begin{pmatrix}   
	\floor{\cos\theta\cdot e^{i\varphi_1}}_{\xi} & \floor{\sin\theta\cdot e^{i\varphi_2}}_{\xi} \\
	\floor{-\sin\theta\cdot e^{-i\varphi_2}}_{\xi} & \floor{\cos\theta\cdot e^{-i\varphi_1}}_{\xi}
	\end{pmatrix}
\]

\end{defn}

Now, using the notation $A_i=U_i\cdot U_{i-1}\cdots U_1 \cdot \vert0^n\rangle \langle 0^n \vert\cdot G_T\cdots G_{i+1}$ we present the full protocol $\widehat{W}$, in which a $\BPP$ verifier $V$ is given a sequence of local quantum gates $g_T,\dots,g_1$, receives a value $C$ from the prover $P$, and for the \textit{top row matrix} $A_{2^n\times 2^n}=\vert0^n\rangle \langle 0^n \vert\cdot G_T\cdot G_{T-1}\cdots G_1$ verifies whether $tr(A)=C$ or $\abs{tr(A)-C} \geq\frac{1}{10\cdot 2^n}=K$ by interacting with $P$. By Claim \ref{cl:postBQPtop}, such a protocol is sufficient to derive an $\IP$ for $\PsP$.

\begin{framed}
	\paragraph{The protocol $\widehat{W}$:} \label{pr:boundedW}

	\begin{enumerate}

		\item {\em In the 0'th round} -- $V$ asks for $M_0=A|_{g_1}$, receives a matrix $M'_0$ from $P$, and verifies that $\left\vert C-tr(M'_0)\right\vert\;\leq\mu$ (rejects otherwise).
  
  		 \item {\em In the $i$'th round} -- $V$ chooses $u_i^1, u_i^2, u_i^3 \backsim \widehat{U}(2)$, sets $u_i = u_i^1\otimes u_i^2 \otimes u_i^3$ on the qubits on which $g_i$ operates, asks for $M_i=A_i|_{g_{i+1}}$,		 receives a matrix $M'_i$, and verifies that $\left\vert tr(M'_i)-tr(M'_{i-1}\cdot \widehat{g_i^{-1}}\cdot \widehat{u_i}\right\vert\leq\mu$ (rejects otherwise, or if $M'_i$ has an entry greater than $2^n$. This should never be the case as each entry is the sum of less that $2^n$ values which are at most 1 each)

		\item {\em In the T'th round} -- $V$ chooses $u_T$ as before, and accepts if: 
		
		$\left\vert tr(\widehat{U}_T\cdot \widehat{U}_{T-1}\cdots \widehat{U}_1 \cdot \vert0^n\rangle \langle 0^n \vert)-tr(M'_{T-1}\cdot \widehat{g_T^{-1}}\cdot \widehat{u}_T)\right\vert\leq\mu$. (rejects otherwise).
	\end{enumerate}
\end{framed}
	
\subsection{Completeness}
\label{ssec:bounded comp}
\begin{claim}[Completeness for the bounded case]\label{cl:bounded completeness}
	if $\vert tr(A)-C\vert \leq \xi=\frac{\mu}{2^{2n+11}T}$, a $\pBQP$ prover $P$ can pass $\widehat{W}$ with probability 1, using $Poly(n)$ bits of communication.
	\end{claim}

We give here the skeleton of the proof, the proofs of the technical 
lemmas can be found in Appendix \ref{app:C5}.

	\begin{proof}
	
	We first note that for completeness in the bounded precision case, we must
 require $P$ to have the power of a $\pBQP$ machine.	
	 This means that when $P$ is asked to perform any given unitary $u$ on a single qubit (or tensor product of such, of course), it computes an approximation $\tilde{u}$ using its universal set of Hadamard and Toffoli gates s.t: $\left\Vert\tilde{u}-u\right\Vert_{L_2}\leq\xi$, and operates with $\tilde{u}$ instead of $u$. $P$ can do this in $Poly(n)$ time as long as $\frac{1}{\xi}=O(2^{n^{c_1}})$ for some constant $c_1$ (By the Solovay-Kitaev Theorem \cite{NC02}). 
	
We use the convenient notation $A'_i=\tilde{U}_i\cdot \tilde{U}_{i-1}\cdots \tilde{U}_1 \cdot \vert0^n\rangle \langle 0^n \vert\cdot G_T\cdots G_{i+1}$, and note that a $\pBQP$ machine $P$ can be used to compute $A'_i\vert_s$ (for any set $s$ of qubits of constant size). This is true because the value of each entry in $A'_i\vert_s$ is the difference between the number of paths which have a positive and negative contribution to it, times $\frac{1}{\sqrt{2}^t}$ (where $t$ is the number of Hadamard gates), and as such can be computed by 2 $\sP$ calls. We now remember $\PP=\pBQP$ \cite{Aar05}, and that for any function $\func \in \sP$ the language $\lang=\left\{(x,k)\vert \func(x)\geq k\right\}$ is in 
$\PP$ \cite{ZOO}\footnote{By taking $f$ for each entry in $A'_i\vert_s$ to be the number of paths which contribute to that entry (once with a positive and once with a negative value), a $\pBQP$ machine can indeed be used to compute $A'_i\vert_s$}. 

We can now define $P$'s \textit{Truthful Strategy} under $\widehat{W}$.

		\begin{defn} $P$'s \textit{Truthful Strategy} under $\widehat{W}$:
		\label{def:truthful}
		\begin{enumerate}

			\item {\em In the 0'th round} -- $P$ sends back $M'_0=A|_{g_1}$.
  
	 		\item {\em In the $i$'th round} -- $P$ receives $u_i$ (represented by its parameters) from $V$, and sends back $M'_i=A'_i\vert_{g_{i+1}}$.

		\end{enumerate}
		
		(Since the prover in this protocol always needs to send an integer multiple of $\frac{1}{\sqrt{2}^k}$, we assume he just sends this integer along with $k$)
	\end{defn}
	Having defined $P$'s \textit{Truthful Strategy} under $\widehat{W}$, we show that if $Tr(A)=C$, $P$ passes $\widehat{W}$ with probability 1 by using its \textit{Truthful Strategy}. To do this we  set $n'=3$ the number of qubits on which each of the gates $g_i$ operates, and use the following claims $\forall_{0\leq i \leq T}$:

	\begin{subclaim}\label{cl:tr(A')-tr(A)}
	$\left\vert tr\left( A'_i\right)-tr\left( A_i\right)\right\vert\leq 2^n T\cdot \xi$
	\end{subclaim}

	\begin{subclaim}\label{cl:tr(M*h)-tr(M'*h)}
	$\left\vert tr\left( M_{i-1}\cdot g_i^{-1}\cdot u_i)\right)-tr\left( M'_{i-1}\cdot g_i^{-1}\cdot u_i)\right)\right\vert\leq 2^{2n'} \cdot 2^n T \cdot \xi$
	\end{subclaim}
	
	\begin{subclaim}\label{cl:tr(M'*h)-tr(M'*h')}
	$\left\vert tr\left( M'_{i-1}\cdot g_i^{-1}\cdot u_i\right)-tr\left( M'_{i-1}\cdot \widehat{g_i^{-1}}\cdot \widehat{u}_i\right) \right\vert\leq  2\cdot 2^{3n'}\cdot 2^n \cdot \xi$
	\end{subclaim}
		
	\begin{subclaim}\label{cl:tr(U'*0)-tr(U*0)}
	$\left\vert tr(\widehat{U}_T\cdot \widehat{U}_{T-1}\cdots \widehat{U}_1 \cdot \vert0^n\rangle \langle 0^n \vert) - tr(U_T\cdot U_{T-1}\cdots U_1 \cdot \vert0^n\rangle \langle 0^n \vert)\right\vert\leq 2^{n}T\cdot\xi$
	\end{subclaim}
	
	The Claims \ref{cl:tr(A')-tr(A)}, \ref{cl:tr(M*h)-tr(M'*h)}, \ref{cl:tr(M'*h)-tr(M'*h')}, \ref{cl:tr(U'*0)-tr(U*0)} are proven in Appendix \ref{app:C5}. \\
	
Now we are ready to wrap up. By using $\xi=\frac{\mu}{2^{2n+11}T} < \frac{\mu}{2\cdot 2^{2n+3n'}(T+2)}$, and remembering that $tr(M_i)=tr(A_i)$, $tr(M'_i)=tr(A'_i)$, $tr(M_i)=tr(M_{i-1}\cdot h_i)$, $tr(U_T\cdot U_{T-1}\cdots U_1 \cdot \vert0^n\rangle \langle 0^n \vert)=tr(M_{T-1}\cdot h_T)$ we get:
	\begin{enumerate}
	
		\item {\em P passes the 0'th round, as} $\left\vert C-tr(M'_0)\right\vert=0<\mu$.

  		 \item {\em P passes the i'th round for $1\leq i < T$, as:}
  		 \begin{align*}
  		 \left\vert tr(M'_i)-tr(M'_{i-1}\cdot \widehat{g_i^{-1}}\cdot \widehat{u}_i)\right\vert
  		 &\leq\left\vert tr(M'_i)-tr(M_i)\right\vert 
  		 + \left\vert tr(M_{i-1}\cdot g_i^{-1}\cdot u_i)-tr(M'_{i-1}\cdot g_i^{-1}\cdot u_i\right\vert \\
  		 &+ \left\vert tr(M'_{i-1}\cdot g_i^{-1}\cdot u_i)-tr(M'_{i-1}\cdot \widehat{g_i{-1}}\cdot \widehat{u}_i)\right\vert 
  		 \leq 2^n T\cdot \xi + 2^{2n'}\cdot 2^n T\cdot\xi \\
  		 &+  2\cdot 2^{3n'}\cdot 2^n \cdot \xi 
  		 \leq 2\cdot 2^{n+3n'}(T+1)\cdot\xi\leq 1024\cdot2^n(T+1) \leq \mu
  		 \end{align*}
  		 and of course, each entry of $M'_i$ is at most the sum of $2^n$ entries of $A_i$, and each entry of $A_i$ is at most $1$ in absolute value, 
  		 so all entries of $M'_i$ are at most $2^n$.

		\item {\em P passes the T'th round, as:} 
		\begin{align*}
		&\left\vert tr(\widehat{U}_T\cdot \widehat{U}_{T-1}\cdots \widehat{U}_1 \cdot \vert0^n\rangle \langle 0^n \vert)-tr(M'_{T-1}\cdot \widehat{g_T^{-1}}\cdot \widehat{u}_T)\right\vert \\
		&\leq \left\vert tr(\widehat{U}_T\cdot \widehat{U}_{T-1}\cdots \widehat{U}_1 \cdot \vert0^n\rangle \langle 0^n \vert) - tr(U_T\cdot U_{T-1}\cdots U_1 \cdot \vert0^n\rangle \langle 0^n \vert \right\vert \\
		&+ \left\vert tr(M_{T-1}\cdot g_T^{-1}\cdot u_T)-tr(M'_{T-1}\cdot g_T^{-1}\cdot u_T\right\vert 
  		+ \left\vert tr(M'_{T-1}\cdot g_T^{-1}\cdot u_T)-tr(M'_{T-1}\cdot \widehat{g_T^{-1}}\cdot \widehat{u}_T\right\vert \\
  		&\leq 2^{n}T\cdot\xi + 2^{2n'}\cdot 2^n T\cdot \xi +  2\cdot 2^{3n'}\cdot 2^n \cdot \xi \\
  		&\leq 2\cdot 2^{n+3n'}(T+1)\cdot\xi\leq 1024\cdot 2^{n}(T+1)
  		 \cdot\xi\leq \mu
		\end{align*}

	\end{enumerate}
	
	And the communication complexity follows trivially by the fact that there is a total of $poly(n)$ rounds, in which a total of $poly(n)$ values (degrees and matrix entries) are being communicated, and for each value it takes $poly(n)$ classical bits to express it's accuracy $\xi$.	
	\end{proof}

\subsection{Soundness}
\begin{claim}[Soundness for the bounded case]\label{cl:bounded soundness}
if $|tr(A)-C|\geq K= \frac{1}{10\cdot 2^n}$, $P$ passes $\widehat{W}$ with probability $\leq\frac{1}{3}$.
\end{claim}

Using the same notion for $\Delta_i$ as in \ref{sec: w soundness}, we will show that given $|tr(A)-C|\geq K = \frac{1}{10\cdot 2^n}$, and conditioned on $P$ passing the first $T-1$ rounds, there is only a small probability for $tr\left(\Delta_{T-1}\right)$ to be small enough for $P$ to be likely to also pass the final $T$'th round.

To do this, we introduce an approximate version of Lemma \ref{lm:basic}, which we use in a similar fashion to the way we used 
Lemma \ref{lm:basic} in proving the soundness for $W$.
\begin{lemma}
\label{lm:bounded}
Let $\Delta$ be an operator on ${n'}$ qubits, and let $u= u^1\otimes \cdots \otimes  u^{n'}$:
\[
\forall K\geq0, m\geq2:\Vert\Delta\Vert_{Frob}\geq K \Rightarrow Pr_{u^1,\cdots, u^{n'} \backsim U(2)}\left(\left\vert tr(\Delta \cdot u)\right\vert<\frac{4K}{\left(16m^3\right)^{n'}}\right)\leq \frac{5{n'}}{m}
\]
\end{lemma}

The proof of Lemma \ref{lm:bounded} can be found in Appendix \ref{app:L3}.
This Lemma \ref{lm:bounded} will be used shortly 
with $n'=3$ and $m=60T$ to bound the probability that 
$\tr{\Delta \cdot u}$ for a randomly chosen unitary $u$ is too small. 
However, before doing this, we have to take care of another 
issue which arises due to precision matters. 
In $\widehat{W}$ the unitaries are not chosen according to $U(2)$, 
the distribution assumed in Lemma \ref{lm:bounded}, but instead they 
are chosen from the distribution $\widehat{U}(2)$.  
Therefore, we also state Claim \ref{cl:prob. dist} which 
enables to quantify the similarity of these two distributions and is proven in Appendix \ref{app:C6}. 
To present Claim \ref{cl:prob. dist}, we define the two following probability distributions:
\begin{defn} We define the distribution of unitaries on $n'$ qubits, $D(n')$, to be the one which is used in protocol $W$. Namely, choosing a unitary operator $u\backsim D(n')$ means choosing $n'$ unitaries on a single qubit $u^1,u^2,\cdots, u^{n'}\backsim U(2)$, and setting $u=u^1\otimes u^2\otimes\cdots\otimes u^{n'}$.
\end{defn}

\begin{defn} We define the distribution of unitaries on $n'$ qubits, $\widehat{D}(n')$, to be the one which is used in protocol $\widehat{W}$. Namely, choosing a unitary operator $u\backsim D(n')$ means choosing $n'$ unitaries on a single qubit $u^1,u^2,\cdots, u^{n'}\backsim \widehat{U}(2)$, and setting $u=u^1\otimes u^2\otimes\cdots\otimes u^{n'}$.
\end{defn}

\begin{subclaim}\label{cl:prob. dist}
let $\Delta$ an operator on $n$ qubits:
\[
Pr_{u\backsim \widehat{D}({n'})}\Big(\left\vert tr\left(\Delta\cdot u \right)\right\vert < \delta-2^{{n'}}{n'}\cdot6\cdot\xi\cdot\Vert\Delta\Vert_{Frob} \Big)
\leq Pr_{u\backsim D({n'})}\Big(\left\vert tr\left(\Delta\cdot u \right)\right\vert < \delta \Big)+3{n'}\frac{\xi}{2\pi}
\]
\end{subclaim}

We can now state the main lemma which will be 
used to prove Claim \ref{cl:bounded soundness}: 
this is Claim \ref{cl:soundness by induction}, 
which shows that there is only a small probability for 
$\Delta_{T-1}$ to be small. 

\begin{subclaim}\label{cl:soundness by induction}
Let $\chi=60^{12}T^9$,
and let $4\chi^T\cdot\mu\leq K$. Suppose $\abs{\tr{A}-C} \geq K$. Then, conditioned on $P$ passing the 
first $i<T$ rounds in $\widehat{W}$: 
\[
Pr_{u_1, u_2, \cdots, u_i \backsim \widehat{D}(n')}\Big(\left\vert tr\left(\Delta_i\right)\right\vert<\frac{K}{4\chi^i} \Big)\leq i\cdot \left(\frac{1}{4T} + 3{n'}\frac{\xi}{2\pi}\right)
\]
\end{subclaim}

The proof uses both Lemma \ref{lm:bounded} 
and the fact that the two distributions are similar, 
and can be found in Subsection \ref{sec:soundness by induction proof}.  
Now we can prove Claim \ref{cl:bounded soundness}. 

{~}

\begin{proof} {\bf Of Claim \ref{cl:bounded soundness}}. 
To pass $\widehat{W}$, $P$ must pass all rounds $i<T$, and then pass the $T$'th round.
	By Claim \ref{cl:soundness by induction}, passing all rounds $i<T$ implies:
	\[
	Pr_{u_1,u_2,\cdots,u_{T-1}\backsim\widehat{D}(n')}\Big(\left\vert tr\left( \Delta_{T-1}\right)\right\vert < \frac{K}{4\chi^{T-1}}\Big)\leq (T-1)\cdot \left(\frac{1}{4T}+3{n'}\frac{\xi}{2\pi}\right)
	\]
	So, denoting the probability that $P$ passes $\widehat{W}$ by $Pr\Big( Success\Big)$, we have:
		
	\begin{align*}\label{10}
	Pr\Big( Success\Big)
	&\leq Pr_{u_1,u_2,\cdots,u_{T-1}\backsim\widehat{D}(n')}\Big(\left\vert tr\left( \Delta_{T-1}\right)\right\vert < \frac{K}{4\chi^{T-1}}\Big)\\ 
	&+ Pr_{
	\substack{u_T \backsim\widehat{D}(n')\\
	\left\vert tr\left( \Delta_{T-1}\right)\right\vert \geq \frac{K}{4\chi^{T-1}}}}
	\bigg(\left\vert tr\left(\widehat{U}_T\cdot \widehat{U}_{T-1}\cdots \widehat{U}_1 \cdot \vert0^n\rangle \langle 0^n \vert\right)-tr\left(M'_{T-1}\cdot \widehat{g_T^{-1}}\cdot \widehat{u}_T\right)\right\vert\leq\mu\bigg)\\
	\end{align*}
	
	To evaluate the second term, we review the proof of Claim \ref{cl:w round T} to note:
	\begin{align*}
	\left\vert tr\left(\Delta_{T-1}\cdot g_T^{-1}\cdot u_T \right)\right\vert
	&=\left\vert tr\left(\left(M_{T-1}-M'_{T-1}\right)\cdot g_T^{-1}\cdot u_T\right)\right\vert
	=\Big\vert tr\left(U_T\cdot U_{T-1}\cdots U_1 \cdot \vert0^n\rangle \langle 0^n \vert - \widehat{U}_T\cdot \widehat{U}_{T-1}\cdots \widehat{U}_1 \cdot \vert0^n\rangle \langle 0^n \vert\right)\\
	&+ tr\left(\widehat{U}_T\cdot \widehat{U}_{T-1}\cdots \widehat{U}_1 \cdot \vert0^n\rangle \langle 0^n \vert\right)-tr\left(M'_{T-1}\cdot \widehat{g_T^{-1}}\cdot \widehat{u}_T\right)
	+tr\left(M'_{T-1}\cdot \widehat{g_T^{-1}}\cdot \widehat{u}_T-M'_{T-1}\cdot g_T^{-1}\cdot u_T\right)\Big\vert\\
&	\leq 2^{2n}T\cdot\xi+2\cdot2^{3n'}\cdot2^n\cdot\xi+
	\left\vert tr\left(\widehat{U}_T\cdot \widehat{U}_{T-1}\cdots \widehat{U}_1 \cdot \vert0^n\rangle \langle 0^n \vert\right)-tr\left(M'_{T-1}\cdot \widehat{g_T^{-1}}\cdot \widehat{u}_T\right)\right\vert 
	\end{align*}	
	
	Where we used a triangle inequality and Claims \ref{cl:tr(U'*0)-tr(U*0)} and \ref{cl:tr(M'*h)-tr(M'*h')} for the inequality (noting the only assumption made on $M_{i-1}$ for Claim \ref{cl:tr(M'*h)-tr(M'*h')} was its entries being smaller than $2^n$ so as to pass $i$'th round test) 
We can now bound $Pr\Big( Success\Big)$ by:
\begin{align*}
Pr\Big( Success\Big)
	&\leq (T-1)\cdot \left(\frac{1}{4T}+3{n'}\frac{\xi}{2\pi}\right) \\
	&+ Pr_{
	\substack{u_{T}\backsim\widehat{D}(n')\\
	\left\vert tr\left( \Delta_{T-1}\right)\right\vert \geq \frac{K}{4\chi^{T-1}}}}
	\bigg(\left\vert tr\left(\Delta_{T-1}\cdot g_T^{-1}\cdot u_T \right)\right\vert
	\leq 2^{2n}T\cdot\xi+2\cdot2^{3n'}\cdot2^n\cdot\xi+\mu \bigg)
\end{align*}

It holds that $\frac{K}{\chi^T} > 2^{2n}T\cdot\xi+2\cdot2^{3n'}\cdot2^n\cdot\xi+\mu + 2^{n'}n'\cdot6\cdot\xi\cdot2^{2n'}\cdot2^n$, and so:

\begin{align*}
&Pr_{
	\substack{u_{T}\backsim\widehat{D}(n')\\
	\left\vert tr\left( \Delta_{T-1}\right)\right\vert \geq \frac{K}{4\chi^{T-1}}}}
	\bigg(\left\vert tr\left(\Delta_{T-1}\cdot g_T^{-1}\cdot u_T \right)\right\vert
	\leq 2^{2n}T\cdot\xi+2\cdot2^{3n'}\cdot2^n\cdot\xi+\mu \bigg)\\
	\leq &Pr_{
	\substack{u_{T}\backsim\widehat{D}(n')\\
	\left\vert tr\left( \Delta_{T-1}\right)\right\vert \geq \frac{K}{4\chi^{T-1}}}}
	\bigg(\left\vert tr\left(\Delta_{T-1}\cdot g_T^{-1}\cdot u_T \right)\right\vert
	\leq \frac{K}{\chi^T}-2^{n'}n'\cdot6\cdot\xi\cdot2^{2n'}\cdot2^n \bigg)\\
	\leq&Pr_{
	\substack{u_{T} \backsim D(n') \\
	\left\vert tr\left(\Delta_{T-1}\right)\right\vert\geq \frac{K}{4\chi^{T-1}}
	}}
	\Big(\left\vert tr\left(\Delta_{T-1}\cdot g_{T}^{-1}\cdot u_{T} \right)\right\vert < \frac{K}{\chi^{T}}\Big)
	+3{n'}\frac{\xi}{2\pi}
	\leq \frac{1}{4T}+3{n'}\frac{\xi}{2\pi}
\end{align*}

Where the second inequality is by Claim \ref{cl:prob. dist}, and the last inequality is by using Inequality \refeq{eq:P2}
(in the proof of Claim \ref{cl:soundness by induction}) for $i=T-1$. \\

This gives us $Pr\Big( Success\Big)\leq T\cdot \left(\frac{1}{4T}+3{n'}\frac{\xi}{2\pi}\right)<\frac{1}{3}$ as required.
\end{proof}

\subsection{Proof of Claim \ref{cl:soundness by induction}}
\label{sec:soundness by induction proof}
\paragraph{Claim \ref{cl:soundness by induction}}
Let $\chi=60^{12}T^9$, and let $4\chi^T\cdot\mu\leq K$. Suppose $\abs{\tr{A}-C} \geq K$. Then, 
conditioned on $P$ passing the first $i<T$ rounds in $\widehat{W}$: 
\[
Pr_{u_1, u_2, \cdots, u_i \backsim \widehat{D}(n')}\Big(\left\vert tr\left(\Delta_i\right)\right\vert<\frac{K}{4\chi^i} \Big)\leq i\cdot \left(\frac{1}{4T} + 3{n'}\frac{\xi}{2\pi}\right)
\]

	\begin{proof}
	By finite induction on $0\leq i\leq T-1$:
	\begin{itemize}
	\item {\em Induction base:} For $i=0$, the condition on passing the $0$'th round gives us $\left\vert C- tr\left(M'_0\right)\right\vert\leq\mu$ and so we have:
	\begin{align*}
	K&\leq\left\vert tr(A)-C\right\vert=\left\vert C-tr(M_0)\right\vert\leq\left\vert C-tr(M'_0)\right\vert+\left\vert tr(M'_0)-tr(M_0)\right\vert\leq\mu+\left\vert tr(\Delta_0)\right\vert \\
	&\Rightarrow\left\vert tr(\Delta_0)\right\vert\geq K-\mu\geq \frac{K}{4}
	\end{align*}
	So we have $Pr\Big(\left\vert tr(\Delta_0)\right\vert< \frac{K}{4}\Big)=0$ as required.

	\item {\em Induction step:} Assuming correctness for $1\leq i<T-2$, we prove for $i+1$. We use the total probability equation to bound the required probability by the probability of 3 separate events:
	\begin{align*}
	&Pr_{u_1, u_2, \cdots, u_{i+1} \backsim \widehat{D}(n')}\Big(\left\vert tr\left(\Delta_{i+1}\right)\right\vert < \frac{K}{4\chi^{i+1}}\Big) \\
	&\leq \underbrace{  Pr_{u_1, u_2, \cdots, u_{i+1} \backsim \widehat{D}(n')}\Big(\left\vert tr\left(\Delta_{i}\right)\right\vert< \frac{K}{4\chi^i} \Big)   }_{P_1}
	+ \underbrace{   Pr_{
	\substack{u_1, u_2, \cdots, u_{i+1} \backsim \widehat{D}(n') \\
	\left\vert tr\left(\Delta_{i}\right)\right\vert\geq \frac{K}{4\chi^i}
	}}
	\Big(\left\vert tr\left(\Delta_{i}\cdot g_{i+1}^{-1}\cdot u_{i+1} \right)\right\vert< \frac{K}{2\chi^{i+1}}\Big)   }_{P_2}\\
	&+\underbrace{   Pr_{
	\substack{u_1, u_2, \cdots, u_{i+1} \backsim \widehat{D}(n') \\
	\left\vert tr\left(\Delta_{i}\right)\right\vert\geq\frac{K}{4\chi^i} \\
	\left\vert tr\left(\Delta_{i}\cdot g_{i+1}^{-1}\cdot u_{i+1} \right)\right\vert\geq \frac{K}{2\chi^{i+1}}
	}}
	\Big(\left\vert tr\left(\Delta_{i+1} \right)\right\vert< \frac{K}{4\chi^{i+1}}\Big)    }_{P_3}
	\end{align*}
	We upper bound each of the addends $P_1, P_2, P_3$:
	\begin{itemize}

	\item {\em $P_1$:}
	By the condition on passing all the first $i$ rounds, along with the induction assumption, we have $P_1 \leq i\cdot\left( \frac{1}{4T} + 3{n'}\frac{\xi}{2\pi}\right)$
	
	\item {\em $P_2$:} 
	We note: 
	\[
	\abs{\tr{\Delta_i}} \geq \frac{K}{4\chi^i}\Rightarrow\left\Vert \Delta_i \right\Vert_{Frob}\geq \frac{K}{4\chi^i}\Rightarrow\left\Vert \Delta_i\cdot g_{i+1}^{-1} \right\Vert_{Frob}\geq \frac{K}{4\chi^i}
	\]
	Using the fact that the Frobenius norm upper bounds the trace of a matrix for the first derivation, and the fact that operating on a matrix by a unitary on either of its sides does not affect it's Frobenius norm. We can now use Lemma \ref{lm:bounded} on $\tilde{\Delta}_i=\Delta_i\cdot g_{i+1}^{-1}$
 (with $n'=3$, $m=60T$) to achieve: 
	\begin{align*}
	\left\Vert\tilde{\Delta}_i\right\Vert_{Frob}\geq \frac{K}{4\chi^i}
	\Rightarrow \frac{5\cdot 3}{60T} 
	&\geq Pr_{u_{i+1} \backsim D(n')} \left(\left\vert tr\left(\tilde{\Delta}_i\cdot u_{i+1} \right)\right\vert< \frac{K}{\chi^{i}\cdot\left(16\left(60T\right)^3\right)^3}\right)\\
	&\geq Pr_{u_{i+1} \backsim D(n')} \left(\left\vert tr\left(\tilde{\Delta}_i\cdot u_{i+1} \right)\right\vert< \frac{K}{\chi^{i}\cdot\left(60\left(60T\right)^3\right)^3}\right)	\\
	&=Pr_{u_{i+1} \backsim D(n')} \left(\left\vert tr\left(\tilde{\Delta}_i\cdot u_{i+1} \right)\right\vert< \frac{K}{\chi^{i+1}}\right)	\\
	\end{align*}
	And so:
	\begin{equation}\label{eq:tr(delta_i)}
	\abs{\tr{\Delta_i}} \geq \frac{K}{4\chi^i} \Rightarrow Pr_{u_{i+1} \backsim D(n')} \left(\left\vert tr\left(\tilde{\Delta}_i\cdot u_{i+1} \right)\right\vert< \frac{K}{\chi^{i+1}}\right)
	\leq \frac{1}{4T}	
	\end{equation}
	
	By the fact that $\Delta_i$ has $2^{2{n'}}$ entries, all of which at most $2^n$, we have  
\[
\Vert\tilde{\Delta}_i\Vert_{Frob}=\Vert\Delta_i\Vert_{Frob}\leq 2^{2n'}\cdot2^n
\]
and so:
\begin{equation}\label{eq:frac K chi i}
2^{n'}{n'}\cdot6\cdot\xi\cdot\Vert\tilde{\Delta}_i\Vert_{Frob}+\frac{K}{2\chi^{i+1}}
\leq 2^{n'}n'\cdot6\cdot\xi\cdot2^{2n'}\cdot2^n+\frac{K}{2\chi^{i+1}}
\leq \frac{K}{\chi^{i+1}}
\end{equation}

	Using Claim \ref{cl:prob. dist}, we have:
	\begin{equation}\label{eq:claim 6.1 usage}
		P_2\leq Pr_{
	\substack{u_{i+1} \backsim D(n') \\
	\left\vert tr\left(\Delta_{i}\right)\right\vert\geq \frac{K}{4\chi^i}
	}}
	\Big(\left\vert tr\left(\Delta_{i}\cdot g_{i+1}^{-1}\cdot u_{i+1} \right)\right\vert < \frac{K}{2\chi^{i+1}} + 2^{n'}{n'}\cdot6\cdot\xi\cdot\Vert\Delta_{i}\cdot g_{i+1}^{-1}\Vert_{Frob}\Big)
	+3{n'}\frac{\xi}{2\pi}
	\end{equation}
	
	And we can now evaluate:
	
	\begin{equation}\label{eq:P2}
	P_2\leq Pr_{
	\substack{u_{i+1} \backsim D(n') \\
	\left\vert tr\left(\Delta_{i}\right)\right\vert\geq \frac{K}{4\chi^i}
	}}
	\Big(\left\vert tr\left(\Delta_{i}\cdot g_{i+1}^{-1}\cdot u_{i+1} \right)\right\vert < \frac{K}{\chi^{i+1}}\Big)
	+3{n'}\frac{\xi}{2\pi}	
	\leq \frac{1}{4T}+3{n'}\frac{\xi}{2\pi}
	\end{equation}
	
Where the first inequality follows from Inequalities \refeq{eq:claim 6.1 usage} and \refeq{eq:frac K chi i}, and the second inequality from Inequality \refeq{eq:tr(delta_i)}.
	
	\item {\em $P_3$:}
	Using the condition $\left\vert tr\left(\Delta_i\cdot h_{i+1} \right) \right\vert\geq \frac{K}{2\chi^{i+1}}$ we have:
	\begin{align*}
	\frac{K}{2\chi^{i+1}} &\leq \left\vert tr\left(\Delta_i\cdot h_{i+1} \right) \right\vert
	=\left\vert tr\left(M_i\cdot h_{i+1} \right) - tr\left(M'_i\cdot h_{i+1} \right)\right\vert
	=\left\vert tr\left(M_{i+1} \right) - tr\left(M'_i\cdot h_{i+1} \right)\right\vert \\
	&=\left\vert tr\left(M_{i+1} \right)- tr\left(M'_{i+1}\right) + tr\left(M'_{i+1} \right) - tr\left(M'_i\cdot h_{i+1} \right)\right\vert \\
	&=\left\vert tr\left(\Delta_{i+1} \right) + tr\left(M'_{i+1} \right) -tr\left(M'_i\cdot \widehat{h}_{i+1} \right)+tr\left(M'_i\cdot \widehat{h}_{i+1} \right) - tr\left(M'_i\cdot h_{i+1} \right)\right\vert \\
	&\leq \left\vert tr\left(\Delta_{i+1} \right) \right\vert
	+ \mu + 2\cdot 2^{3n'}\cdot 2^n \cdot \xi
	\end{align*}
	\begin{equation}
	\Rightarrow \left\vert tr\left(\Delta_{i+1} \right) \right\vert
	\geq \frac{K}{2\chi^{i+1}} - \mu - 2\cdot 2^{3n'}\cdot 2^n \cdot \xi \geq \frac{K}{4\chi^{i+1}}
	\Rightarrow P_3 = 0
	\end{equation}

	\end{itemize}
	Where we used a triangle inequality along with the condition on passing the $(i+1)$'th round and Claim \ref{cl:tr(M'*h)-tr(M'*h')} to achieve the inequality.

	Summing the 3 addends, we get:
	\[
	P_1 + P_2 + P_3 \leq \frac{1}{4T} +3{n'}\frac{\xi}{2\pi} +i\cdot \left(\frac{1}{4T}+3{n'}\frac{\xi}{2\pi}\right)= (i+1)\cdot \left(\frac{1}{4T}+3{n'}\frac{\xi}{2\pi}\right)
	\]
	As required, thus concluding the proof for the induction step and hence of Claim \ref{cl:soundness by induction} as well.	
	\end{itemize}

	\end{proof}

\section{Acknowledgments}
Both authors acknowledge the generous support of ERC grant number 280157.\\
We also thank Or Sattath and Thomas Vidick for helpful discussions. 

\bibliographystyle{alpha}
\bibliography{bib}

\appendix
\clearpage
\section{Lemma \ref{lm:basic} proof}
\label{app:L1}
\paragraph{Lemma \ref{lm:basic}.}

Let $\Delta$ an operator on $n$ qubits and let $u=u^1\otimes\cdots\otimes u^n$ a unitary operator on $n$ qubits:
\[
\Delta\neq0 \Rightarrow Pr_{u^1,\cdots,u^n\backsim U(2)}\Big(tr(\Delta\cdot u)=0\Big)=0
\]

	\begin{proof} By induction on $n$:

		\begin{itemize}
			\item {\em \ref{lm:basic}.1 Induction base:} For $n=1$ we have:
		
			\[
			u = u^1 =  \begin{pmatrix}   
				\cos\theta\cdot e^{i\varphi_1} & \sin\theta\cdot e^{i\varphi_2} \\
				-\sin\theta\cdot e^{-i\varphi_2} & \cos\theta\cdot e^{-i\varphi_1}
			\end{pmatrix},
			\Delta = \begin{pmatrix}   
				a & c \\
				d & b
			\end{pmatrix},	
			\]
			and we can assume $a \neq 0$, as the following derivation is similar regardless of which entry of $\Delta$ is non-zero. 

			So we have:
			\begin{align*}
			&Pr_{u \backsim U(2)}\Big(tr(\Delta \cdot u)=0\Big)\\
			&=Pr_{\theta, \varphi_1, \varphi_2 \backsim [0,2\pi]}\Big(a\cdot \cos\theta \cdot e^{i\varphi_1}+d\cdot \sin\theta \cdot e^{i\varphi_2}-c\cdot \sin\theta \cdot e^{-i\varphi_2}+b\cdot \cos\theta \cdot e^{-i\varphi_1}=0\Big)\\			
			&=Pr_{\theta, \varphi_1, \varphi_2 \backsim [0,2\pi]}\Big(\cos\theta \left(a \cdot e^{i\varphi_1} + b \cdot e^{-i\varphi_1}\right) 
			+ \sin\theta \left(d \cdot e^{i\varphi_2}-c\cdot e^{-i\varphi_2}\right)=0\Big)\\		
			&\leq Pr_{\theta\backsim [0,2\pi]}\Big(\sin\theta=0\Big)+
			Pr_{\theta\backsim [0,2\pi]}\Big(\cos\theta=0\Big)+
			Pr_{\varphi_1\backsim [0,2\pi]}\Big(a \cdot e^{i\varphi_1} + b \cdot e^{-i\varphi_1}=0\Big)+
			\\ 
			&+Pr_{\substack{a \cdot e^{i\varphi_1} + b \cdot e^{-i\varphi_1} \neq0\\
			\cos\theta\neq0\\
			\sin\theta\neq0
			}}
			\Big(\left\vert 
			\cos\theta \left(a \cdot e^{i\varphi_1} + b \cdot e^{-i\varphi_1}\right)\right\vert = \left\vert \sin\theta \left(d \cdot e^{i\varphi_2}-c\cdot e^{-i\varphi_2}\right)
			\right\vert\Big)
			\end{align*}
			we calculate an upper bound on the sum by proving an upper bound on each of the probabilities separately: 
			\begin{enumerate}
				\item{$Pr_{\theta\backsim [0,2\pi]}\Big(\sin\theta=0\Big)=0$} 		
				
				\item{$Pr_{\theta\backsim [0,2\pi]}\Big(\cos\theta=0\Big)=0$} 
								
				\item{To see $Pr_{\varphi_1\backsim [0,2\pi]}\Big(a \cdot e^{i\varphi_1} + b \cdot e^{-i\varphi_1}=0\Big)=0$ we observe:}
				\begin{align*}
				Pr_{\varphi_1\backsim [0,2\pi]}&\Big(a \cdot e^{i\varphi_1} + b \cdot e^{-i\varphi_1}=0\Big)\\
				&=Pr_{\varphi_1\backsim [0,2\pi]}\Big(a \left( \cos\varphi_1 + i\cdot \sin \varphi_1\right) +b\left(\cos \varphi_1- i\cdot \sin\varphi_1\right)=0\Big)\\
				&=Pr_{\varphi_1\backsim [0,2\pi]}
				\left(\sqrt{
				\left((a+b)\cos\varphi_1\right)^2 + \left((a-b)\sin \varphi_1\right)^2
				}=0\right)
				\end{align*}
				
				we assume $b$ and $a$ have the same sign, and so:
				\begin{align*}
				&Pr_{\varphi_1\backsim [0,2\pi]}
				\left(\sqrt{
				\left((a+b)\cos\varphi_1\right)^2 + \left((a-b)\sin \varphi_1\right)^2
				}=0\right)\\
				\leq&Pr_{\varphi_1\backsim [0,2\pi]}
				\left(\sqrt{
				\left((a+b)\cos\varphi_1\right)^2
				}=0\right)
				= Pr_{\varphi_1\backsim [0,2\pi]}\Big(\cos\varphi_1=0\Big)=0
				\end{align*}
				Where the final equality comes from the fact  that $a\neq0$ and $a$ has the same sign as $b$. If $a$ and $b$ do not have the same sign - we simply bound using $\sqrt{\left((a-b)\sin\varphi_1\right)^2}$.

				\item{$Pr_{\substack{a \cdot e^{i\varphi_1} + b \cdot e^{-i\varphi_1} \neq0\\
			\cos\theta\neq0\\
			\sin\theta\neq0
			}}
			\Big(\left\vert 
			\cos\theta \left(a \cdot e^{i\varphi_1} + b \cdot e^{-i\varphi_1}\right)\right\vert = \left\vert \sin\theta \left(d \cdot e^{i\varphi_2}-c\cdot e^{-i\varphi_2}\right)
			\right\vert\Big)$}
			\begin{align*}
			=&Pr_{\substack{a \cdot e^{i\varphi_1} + b \cdot e^{-i\varphi_1} \neq0\\
			\cos\theta\neq0\\
			\sin\theta\neq0
			}}
			\Big(\left\vert 
			cot\theta\right\vert \left\vert \left(a \cdot e^{i\varphi_1} + b \cdot e^{-i\varphi_1}\right)\right\vert = \left\vert d \cdot e^{i\varphi_2}-c\cdot e^{-i\varphi_2}
			\right\vert\Big)\\
			=&Pr_{\substack{a \cdot e^{i\varphi_1} + b \cdot e^{-i\varphi_1} \neq0\\
			\cos\theta\neq0\\
			\sin\theta\neq0
			}}				
			\left(\left\vert
			cot(\theta)\right\vert = \underbrace{\left\vert\frac{d \cdot e^{i\varphi_2}-c\cdot e^{-i\varphi_2}}{a \cdot e^{i\varphi_1} + b \cdot e^{-i\varphi_1}}\right\vert}_H\right)
			=0
			\end{align*}
			As H is some constant number independent of $\theta$.
			\end{enumerate}
					
			Summing the 4 probabilities, we get the induction base for Lemma \ref{lm:basic}.	
	
			\item {\em \ref{lm:basic}.2 Induction step:} Assuming correctness for $n$, we prove for $n+1$. Choosing a basis for the $n+1$ qubits where the $(n+1)$'th qubit is the most significant qubit, we have:
			\[
			u^{n+1} = \begin{pmatrix}   
				u_1 & u_3 \\
				u_4 & u_2
			\end{pmatrix},
			u = \begin{pmatrix}   
				u_1 \cdot \begin{pmatrix} U \end{pmatrix}_{n} & u_3 \cdot \begin{pmatrix} U \end{pmatrix}_{n} \\
				u_4 \cdot \begin{pmatrix} U \end{pmatrix}_{n} & u_2 \cdot \begin{pmatrix} U \end{pmatrix}_{n} 
			\end{pmatrix},
			\Delta = \begin{pmatrix}   
				\begin{pmatrix} \Delta_1 \end{pmatrix}_{n} & \begin{pmatrix} \Delta_3 \end{pmatrix}_{n} \\
				\begin{pmatrix} \Delta_4 \end{pmatrix}_{n} & \begin{pmatrix} \Delta_2 \end{pmatrix}_{n}
			\end{pmatrix}
			\]
			Where $(U)_{n}=u^1\otimes u^2 \otimes \cdots \otimes  u^{n}$, and we denote the number of qubits each block operates on in subscript, for clarity. Now we have:

			\begin{align*}
			tr(\Delta\cdot u)
			&=tr\left(\begin{pmatrix}   
				\begin{pmatrix} \Delta_1 \end{pmatrix} & 
				\begin{pmatrix} \Delta_3 \end{pmatrix} \\
				\begin{pmatrix} \Delta_4 \end{pmatrix} & 
				\begin{pmatrix} \Delta_2 \end{pmatrix}
			\end{pmatrix}\cdot
			\begin{pmatrix}   
				u_1 \cdot \begin{pmatrix} U \end{pmatrix} & 
				u_3 \cdot \begin{pmatrix} U \end{pmatrix} \\
				u_4 \cdot \begin{pmatrix} U \end{pmatrix} & 
				u_2 \cdot \begin{pmatrix} U \end{pmatrix} 
			\end{pmatrix}\right)\\
			&=u_1 \cdot tr\begin{pmatrix} \Delta_1 \cdot U \end{pmatrix}+ u_4 \cdot tr\begin{pmatrix} \Delta_3 \cdot U \end{pmatrix} + u_3 \cdot tr\begin{pmatrix} \Delta_4 \cdot U \end{pmatrix}+ u_2 \cdot tr\begin{pmatrix} \Delta_2 \cdot U \end{pmatrix}\\
			&=tr\left(\underbrace{\begin{pmatrix}   
				tr(\Delta_1\cdot U) & tr(\Delta_3\cdot U)\\
				tr(\Delta_4\cdot U) & tr(\Delta_2\cdot U
			\end{pmatrix} \cdot}_{\Delta'}			
			\begin{pmatrix}   
				u_1 & u_3 \\
				u_4 & u_2
			\end{pmatrix}\right)
			\end{align*}
			
			\begin{equation} \label{eq:tr(delta*u)}
			\Rightarrow tr(\Delta\cdot u)=tr(\Delta'\cdot u^{n+1})
			\end{equation}

			W.L.O.G we can assume $\Delta_1\neq0$, so our induction base tells us $tr(\Delta_1 \cdot U)\neq0$, meaning $\Delta'\neq0$. Using the induction base once again on $\Delta',u^{n+1}$, we obtain the required induction step.
			
			
		\end{itemize}

	\end{proof}

\clearpage
\section{Lemma \ref{lm:reduced assosiativeness} proof}
\label{app:L2}

\paragraph{Lemma \ref{lm:reduced assosiativeness}.}
Let $U$ an operator on a hilbert space \(\mathcal{H}\), and $q$ an operator on a subsystem $Q$ of \(\mathcal{H}\), then $U|_Q\cdot q=(U\cdot(q\otimes I_{\overline{Q}}))|_Q$ 	

	\begin{proof}
		By induction on $i$, the number of qubits $Q$ works on non-trivially:
		\begin{itemize}
			\item {\em Induction base:} For $i=1$ we have (choosing a basis for \(\mathcal{H}\) where the single qubit in $Q$ is the most significant qubit):
			\[
				q = \begin{pmatrix}   
					q_1 & q_3 \\
					q_4 & q_2
					\end{pmatrix},
				U = \begin{pmatrix}   
					\begin{pmatrix} U_1 \end{pmatrix} & \begin{pmatrix} U_3 \end{pmatrix} \\
					\begin{pmatrix} U_4 \end{pmatrix} & \begin{pmatrix} U_2 \end{pmatrix}
					\end{pmatrix}		
			\]
			and so:
			\begin{align*}
			(U\cdot(q\otimes I_{\overline{Q}}))|_Q
			&=\left. \left(\begin{pmatrix}   
				\begin{pmatrix} U_1 \end{pmatrix} & \begin{pmatrix} U_3 \end{pmatrix} \\
				\begin{pmatrix} U_4 \end{pmatrix} & \begin{pmatrix} U_2 \end{pmatrix}
				\end{pmatrix}	\cdot
				\begin{pmatrix}   
				q_1\cdot \begin{pmatrix} I_{n-1} \end{pmatrix} & q_3\cdot \begin{pmatrix} I_{n-1} 			\end{pmatrix} \\
				q_4\cdot \begin{pmatrix} I_{n-1} \end{pmatrix} & q_2\cdot \begin{pmatrix} I_{n-1} 			\end{pmatrix}
				\end{pmatrix}	 
			  \right) \right|_Q\\
			&=\left. \begin{pmatrix}   
				\begin{pmatrix} q_1\cdot U_1 + q_4\cdot U_3 \end{pmatrix} & \begin{pmatrix} 				q_3\cdot U_1 + q_2\cdot U_3\end{pmatrix} \\
				\begin{pmatrix} q_1\cdot U_4 + q_4\cdot U_2 \end{pmatrix} & \begin{pmatrix} 				q_3\cdot U_4 + q_2\cdot U_2\end{pmatrix} 
				\end{pmatrix}	 
			  \right|_Q\\	
			&=\begin{pmatrix}   
				tr\begin{pmatrix} q_1\cdot U_1 + q_4\cdot U_3 \end{pmatrix} & tr\begin{pmatrix} q_3\cdot U_1 + q_2\cdot U_3\end{pmatrix} \\
				tr\begin{pmatrix} q_1\cdot U_4 + q_4\cdot U_2 \end{pmatrix} & tr\begin{pmatrix} q_3\cdot U_4 + q_2\cdot U_2\end{pmatrix} 
			  \end{pmatrix}\\
			&=\begin{pmatrix}   
				q_1\cdot tr\begin{pmatrix} U_1 \end{pmatrix} + q_4\cdot tr\begin{pmatrix} U_3 																				\end{pmatrix} 
				& q_3\cdot tr\begin{pmatrix} U_1 \end{pmatrix} + q_2\cdot tr\begin{pmatrix} U_3 																			\end{pmatrix} \\
				q_1\cdot tr\begin{pmatrix} U_4 \end{pmatrix} + q_4\cdot tr\begin{pmatrix} U_2 																				\end{pmatrix} 
				& q_3\cdot tr\begin{pmatrix} U_4 \end{pmatrix} + q_2\cdot tr\begin{pmatrix} U_2 																			\end{pmatrix} 
			  \end{pmatrix}\\
			&=\begin{pmatrix}   
			tr\begin{pmatrix} U_1 \end{pmatrix} & tr\begin{pmatrix} U_3 \end{pmatrix} \\
			tr\begin{pmatrix} U_4 \end{pmatrix} & tr\begin{pmatrix} U_2 \end{pmatrix}
			\end{pmatrix}\cdot 
			\begin{pmatrix}   
			q_1 & q_3 \\
			q_4 & q_2
			\end{pmatrix}\\
			&=U|_Q\cdot q
			\end{align*}
	
			\item {\em Induction step:} Assuming correctness for $i$, we prove for $i+1$. Choosing a basis  for \(\mathcal{H}\) where the $(i+1)$'th qubit in $Q$ is the most significant qubit, and $Q$'s other $i$ qubits are the following significant bits - we have:
			\[
			q = \begin{pmatrix}   
					\begin{pmatrix} q_1 \end{pmatrix}_i & \begin{pmatrix} q_3 \end{pmatrix}_i \\
					\begin{pmatrix} q_4 \end{pmatrix}_i & \begin{pmatrix} q_2 \end{pmatrix}_i
				\end{pmatrix}_{i+1},
			U = \begin{pmatrix}   
					\begin{pmatrix} U_1 \end{pmatrix}_{n-1} & \begin{pmatrix} U_3 \end{pmatrix}_{n-1} \\
					\begin{pmatrix} U_4 \end{pmatrix}_{n-1} & \begin{pmatrix} U_2 \end{pmatrix}_{n-1}
				\end{pmatrix}_{n}		
			\]
			And we denote the number of qubits each operator operates on in subscript, for clarity. Now we have:	
		\begin{align*}
		&(U(q\otimes I_{\overline{Q}}))|_Q =\\
		&=\left. \left(
			\begin{pmatrix}   
				\begin{pmatrix} U_1 \end{pmatrix} & \begin{pmatrix} U_3 \end{pmatrix} \\
				\begin{pmatrix} U_4 \end{pmatrix} & \begin{pmatrix} U_2 \end{pmatrix}
			\end{pmatrix}	\cdot
			\begin{pmatrix}   
				\begin{pmatrix} q_1 \end{pmatrix}\otimes  \begin{pmatrix} I_{n-i-1} \end{pmatrix} & 						\begin{pmatrix} q_3 \end{pmatrix}\otimes \begin{pmatrix} I_{n-i-1} \end{pmatrix} \\
				\begin{pmatrix} q_4 \end{pmatrix}\otimes \begin{pmatrix} I_{n-i-1} \end{pmatrix} & 							\begin{pmatrix} q_2 \end{pmatrix}\otimes \begin{pmatrix} I_{n-i-1} \end{pmatrix}
			\end{pmatrix}	 
		  \right) \right|_Q\\
		&=\left. 
			\begin{pmatrix}   
				\begin{pmatrix} U_1  (q_1 \otimes I_{n-i-1}) + U_3 (q_4 \otimes I_{n-i-1}) \end{pmatrix} & 													\begin{pmatrix} U_1 (q_3 \otimes I_{n-i-1}) + U_3 (q_2 \otimes I_{n-i-1}) \end{pmatrix} \\
				\begin{pmatrix} U_4  (q_1 \otimes I_{n-i-1}) + U_2 (q_4 \otimes I_{n-i-1}) \end{pmatrix} & 													\begin{pmatrix} U_4 (q_3 \otimes I_{n-i-1}) + U_2 (q_2 \otimes I_{n-i-1}) \end{pmatrix}
			\end{pmatrix}	 
		  \right|_Q\\
		&=\begin{pmatrix}   
			\left. \begin{pmatrix} U_1 (q_1 \otimes I_{n-i-1}) + U_3(q_4 \otimes I_{n-i-1}) \end{pmatrix}\right|_{Q \backslash1} &
			\left. \begin{pmatrix} U_1 (q_3 \otimes I_{n-i-1}) + U_3(q_2 \otimes I_{n-i-1})	\end{pmatrix} \right|_{Q \backslash1} \\
			\left. \begin{pmatrix} U_4 (q_1 \otimes I_{n-i-1}) + U_2 (q_4 \otimes I_{n-i-1}) \end{pmatrix}\right|_{Q \backslash1} &
			\left. \begin{pmatrix} U_4(q_3 \otimes I_{n-i-1}) + U_2 (q_2 \otimes I_{n-i-1}) \end{pmatrix} \right|_{Q \backslash1}
		  \end{pmatrix}\\			
		&=\begin{pmatrix}   
			\begin{pmatrix} U_1|_{Q\backslash 1} \cdot q_1 + U_3|_{Q\backslash 1} \cdot q_4 \end{pmatrix} & 
			\begin{pmatrix} U_1|_{Q\backslash 1} \cdot q_3 + U_3|_{Q\backslash 1} \cdot q_2 \end{pmatrix} \\
			\begin{pmatrix} U_4|_{Q\backslash 1} \cdot q_1 + U_2|_{Q\backslash 1} \cdot q_4 			\end{pmatrix} &
			\begin{pmatrix} U_4|_{Q\backslash 1} \cdot q_3 + U_2|_{Q\backslash 1} \cdot q_2 				\end{pmatrix} 
		  \end{pmatrix}\\
		&=\begin{pmatrix}   
			\begin{pmatrix} U_1 \end{pmatrix}|_{Q\backslash1} & 
			\begin{pmatrix} U_3 \end{pmatrix}|_{Q\backslash 1} \\
			\begin{pmatrix} U_4 \end{pmatrix}|_{Q\backslash1} & 
			\begin{pmatrix} U_2 \end{pmatrix}|_{Q\backslash 1}
		\end{pmatrix} \cdot 
		\begin{pmatrix}   
		  \begin{pmatrix} q_1 \end{pmatrix} & \begin{pmatrix} q_3 \end{pmatrix} \\
		  \begin{pmatrix} q_4 \end{pmatrix} & \begin{pmatrix} q_2 \end{pmatrix}
		\end{pmatrix}\\
		&=U|_Q\cdot q
		\end{align*}
	
		Where we used the induction step assumption for the 4'th equality. 
		
		\end{itemize}
		And so, by induction we conclude the proof of Lemma \ref{lm:reduced assosiativeness}.
	
	\end{proof}

\clearpage
\section{Proof of Theorem \ref{thm:maintheorem2} given Theorem \ref{thm:maintheorem1}}
\label{app:T1.2}
\paragraph{Theorem \ref{thm:maintheorem2}}
$\PsP\in \IP[\BPP, \PP]$. 
\begin{proof}
Let $\lang \in \PsP = \PPP = \PpBQP$, then there exists an algorithm $\alg (x)$ that runs in time polynomial in time $\abs{x}$ using $m=\poly{\abs{x}}$ queries to a $\pBQP$ oracle, such that if $x\in\lang$ $\alg (x)=1$ and otherwise $\alg (x)=0$. For $i\in\lbrace1\dots m\rbrace$ we denote the $i$'th query in $\alg (x)$ $q_i$, and its answer by the oracle $r_i$. Each such quesry $q_i$ can be thought of as a request for the truth value of a term $y\in \lang '$ for some $y$ (of size polynomial in $\abs{x}$) and a $\pBQP$ complete language $\lang '$. We denote the query for the truth value of $y \in \overline{\lang '} $ by $\overline{q_i}$. Since $\pBQP$ is closed under complement, $\overline{\lang ' }\in \pBQP$ as well. By Theorem \ref{thm:maintheorem1}, this means that there is an 
interactive protocol $W_i$ between a $\BPP$ verifier $V$ and a prover $P$ with soundness at most $\frac{1}{3}$ and completeness at least $\frac{2}{3}$ for $q_i$, and a similar protocol $\overline{W_i}$ for $\overline{q_i}$, and the completeness holds even if $P$ is restricted to be a $\pBQP$ machine. 
We now consider the following interactive protocol between a $\BPP$ verifier $V$ and a prover $P$:\\
$V$ simulates $\alg$, but whenever $\alg$ calls for a query $q_i$, $V$ perform the following:
\begin{enumerate}
\item $V$ repeats the protocol $W_i$ $\log(m)+2$ times. If all runs results with accepting, sets $r'_i=1$. Otherwise:
\item $V$ repeats the protocol $\overline{W_i}$ $\log(m)+2$ times. If all runs results with accepting, sets $r'_i=0$. Otherwise, $V$ rejects.
\end{enumerate}
$V$ regards the $r'_i$ values as though they were the correct answers for the $\pBQP$ queries, and he then accepts or rejects according to the result of the simulated $\alg$.
Analysing this protocol, we have:
\begin{itemize}
	\item{Completeness:} Let $x \in \lang$, and let $P$ a $\pBQP$ machine. Assume $P$ tries to make $W_i$ accept if and only if  $r_i=1$, and that $P$ tries to make $\overline{W_i}$ accept if and only if $r_i=0$. By the completeness of $W_i$, this means that $r_i=1\Rightarrow r'_i=1$. On the other hand, if $r_i=0$, $P$ deliberately fails his $W_i$ protocol, and then goes through the series of $\overline{W_i}$ protocols. Along with the completeness of $\overline{W_i}$ this means $r_i=0\Rightarrow r'_i=0$, and together we have that $r_i=r'_i$, and so $V$ will accept with probably 1, since $\alg(x)$ accepts.
	\item{Soundness} Let $x \notin \lang$. By the soundness of $W_i$, $\overline{W_i}$ we have:
	\begin{align*}
	\forall i\in\lbrace1\dots m\rbrace: \quad  \prob{r_i\neq r'_i}\leq& \prob{r'_i=1 \vert r_i=0} + \prob{r'_i=0\vert r_i=1}\\
						\leq& (\prob{W_i=1\vert r_i=0})^{\log{m}+2} + (\prob{\overline{W_i}=0\vert r_i=1})^{\log{m}+2} \leq \frac{2}{3^{\log{m}+2}}\\
	\end{align*}
	Now, as we know that $\alg(x)=0$, we can use the union bound to get:
	\begin{align*}
	\prob{V\ Accepts} \leq \prob{\exists i:\ r_i \neq r'_i} \leq m*\prob{r_i\neq r'_i} \leq  \frac{2m}{3^{\log{m}+2}} < \frac{1}{3} 
	\end{align*}
\end{itemize}
Having shown an interactive proof protocol for $\lang$ with soundness less than $\frac{1}{3}$ and completeness 1, even when the prover $P$ is restricted to be a $\pBQP$ machine, this shows $\PsP\in \IP[\BPP, \PP]$, and by  $\PP = \pBQP$ we have Theorem \ref{thm:maintheorem2}.
\end{proof}

\clearpage
\section{Lemma \ref{lm:bounded} proof}
\label{app:L3}

\paragraph{Lemma \ref{lm:bounded}.}

Let $\Delta$ be an operator on ${n'}$ qubits, and let $u= u^1\otimes \cdots \otimes  u^{n'}$:
\[
\forall K\geq0, m\geq2:\Vert\Delta\Vert_{Frob}\geq K \Rightarrow Pr_{u^1,\cdots, u^{n'} \backsim U(2)}\left(\left\vert tr(\Delta \cdot u)\right\vert<\frac{4K}{\left(16m^3\right)^{n'}}\right)\leq \frac{5{n'}}{m}
\]
	\begin{proof} By induction on ${n'}$:

		\begin{itemize}
			\item {\em Induction base:} For ${n'}=1$ we have:
		
			\[
			u = u^1 =  \begin{pmatrix}   
				\cos\theta\cdot e^{i\varphi_1} & \sin\theta\cdot e^{i\varphi_2} \\
				-\sin\theta\cdot e^{-i\varphi_2} & \cos\theta\cdot e^{-i\varphi_1}
			\end{pmatrix},
			\Delta = \begin{pmatrix}   
				a & c \\
				d & b
			\end{pmatrix},	
			\]
			and we can assume $\vert a\vert\ = max(\vert a\vert,\vert b\vert,\vert c\vert,\vert d\vert)\Rightarrow \vert a\vert\geq \frac{K}{2}$, as the following derivation is similar regardless of which entry of $\Delta$ is maximal in absolute value. 
			So we have:
			\begin{align*}
			&Pr_{u \backsim U(2)}\left(\left\vert tr(\Delta \cdot u)\right\vert<\frac{K}{4m^3}\right)\\
			&=Pr_{\theta, \varphi_1, \varphi_2 \backsim [0,2\pi]}\left(\left\vert a\cdot \cos\theta \cdot e^{i\varphi_1}+d\cdot \sin\theta \cdot e^{i\varphi_2}-c\cdot \sin\theta \cdot e^{-i\varphi_2}+b\cdot \cos\theta \cdot e^{-i\varphi_1}\right\vert<\frac{K}{4m^3}\right)\\			
			&=Pr_{\theta, \varphi_1, \varphi_2 \backsim [0,2\pi]}\left(\left\vert 
			\cos\theta \left(a \cdot e^{i\varphi_1} + b \cdot e^{-i\varphi_1}\right) 
			+ \sin\theta \left(d \cdot e^{i\varphi_2}-c\cdot e^{-i\varphi_2}\right)
			\right\vert<\frac{K}{4m^3}\right)\\			
			&\leq Pr_{\theta\backsim [0,2\pi]}\left(\left\vert \sin\theta\right\vert\leq\frac{1}{m}\right)+
			Pr_{\theta\backsim [0,2\pi]}\left(\left\vert \cos\theta\right\vert\leq\frac{1}{m^2}\right)+
			Pr_{\varphi_1\backsim [0,2\pi]}\left(\left\vert a \cdot e^{i\varphi_1} + b \cdot e^{-i\varphi_1}\right\vert\leq\frac{K}{2m}\right)+
			\\ 
			&+Pr_{\substack{\left\vert \cos\theta\right\vert>\frac{1}{m^2}\\
			\left\vert \sin\theta\right\vert>\frac{1}{m}\\
			\left\vert a \cdot e^{i\varphi_1} + b \cdot e^{-i\varphi_1} \right\vert>\frac{K}{2m}}}
			\left(\left\vert\left\vert 
			\cos\theta \left(a \cdot e^{i\varphi_1} + b \cdot e^{-i\varphi_1}\right)\right\vert- \left\vert \sin\theta \left(d \cdot e^{i\varphi_2}-c\cdot e^{-i\varphi_2}\right)
			\right\vert\right\vert<\frac{K}{4m^3}\right)
			\end{align*}
			we calculate an upper bound on the sum by proving an upper bound on each of the probabilities separately: 
			\begin{itemize}
				\item{$Pr_{\theta\backsim [0,2\pi]}\Big(\left\vert \sin\theta\right\vert\leq\frac{1}{m}\Big)\leq\frac{1}{m}$:} 
				We will prove for $\cos\theta$ instead (and the probabilities are equal, as the functions only differ by a cyclic offset in the section):
				\paragraph{Lemma \ref{lm:bounded}.1:}
				$Pr_{\theta\backsim [0,2\pi]}\Big(\left\vert \cos\theta\right\vert\leq\frac{1}{m}\Big)\leq\frac{1}{m}$ 
				\begin{proof}
				Due to symmetry around $\pi$, we have:
				\[
				Pr_{\theta\backsim [0,2\pi]}\left(\left\vert \cos\theta\right\vert\leq\frac{1}{m}\right)=Pr_{\theta\backsim [\pi,2\pi]}\left(\left\vert \cos\theta\right\vert\leq\frac{1}{m}\right)
				\]
				We denote $\theta_{1}=\cos^{-1}\left(\frac{1}{m}\right),\ \theta_{2}=\cos^{-1}\left(-\frac{1}{m}\right),\ \theta_{1,2}\in [\pi,2\pi]$ and as the cosine function is monotonically increasing in $[\pi,2\pi]$ we have:
				\[
				Pr_{\theta\backsim [\pi,2\pi]}\left(\left\vert \cos\theta\right\vert\leq\frac{1}{m}\right)=\abs{\frac{\theta_1-\theta_2}{\pi}}
				\]				
				
				Remembering the lemma assumes $m\geq2$, we have:
				\[
				\theta\in[\pi,2\pi],\left\vert \cos\theta\right\vert\leq\frac{1}{m}\Rightarrow\theta\in[1\frac{1}{3}\pi,1\frac{2}{3}\pi]
				\]
				since the derivative for $\cos\theta$ in this section is at least $\frac{\sqrt{3}}{2}$, we can use The Mean Value Theorem on $\theta_{1,2}$ to achieve $\abs{\theta_1-\theta_2}\leq\frac{\frac{1}{m}-(-\frac{1}{m})}{\frac{\sqrt{3}}{2}}=\frac{4}{\sqrt{3}m}$, and so:
				\[
				Pr_{\theta\backsim [0,2\pi]}\left(\left\vert \cos\theta\right\vert\leq\frac{1}{m}\right)\leq\frac{4}{\sqrt{3}m\cdot\pi}\leq\frac{0.74}{m}\leq\frac{1}{m}
				\]
				\end{proof}
				
				\item{$Pr_{\theta\backsim [0,2\pi]}\Big(\left\vert \cos\theta\right\vert\leq\frac{1}{m^2}\Big)\leq\frac{1}{m}$:} 
				This follows trivially by a variable change to Lemma \ref{lm:bounded}.1.				
								
				\item{To see $Pr_{\varphi_1\backsim [0,2\pi]}\Big(\left\vert a \cdot e^{i\varphi_1} + b \cdot e^{-i\varphi_1}\right\vert\leq\frac{K}{2m}\Big)\leq\frac{1}{m}$ we observe:}
				\begin{align*}
				Pr_{\varphi_1\backsim [0,2\pi]}&\left(\left\vert a \cdot e^{i\varphi_1} + b \cdot e^{-i\varphi_1}\right\vert\leq\frac{K}{2m}\right)\\
				&=Pr_{\varphi_1\backsim [0,2\pi]}\left(\left\vert 
				a \left( \cos\varphi_1 + i\cdot \sin \varphi_1\right) +b\left(\cos \varphi_1- i\cdot \sin\varphi_1\right)\right\vert\leq\frac{K}{2m}\right)\\
				&=Pr_{\varphi_1\backsim [0,2\pi]}
				\left(\sqrt{
				\left((a+b)\cos\varphi_1\right)^2 + \left((a-b)\sin \varphi_1\right)^2
				}\leq\frac{K}{2m}\right)
				\end{align*}
				W.L.O.G we assume $a, b$ have the same sign (otherwise, we bound using the sine) and so:
				\begin{align*}
				&Pr_{\varphi_1\backsim [0,2\pi]}
				\left(\sqrt{
				\left((a+b)\cos\varphi_1\right)^2 + \left((a-b)\sin \varphi_1\right]^2
				}\leq\frac{K}{2m}\right)\\
				\leq&Pr_{\varphi_1\backsim [0,2\pi]}
				\left(\sqrt{
				\left[(a+0)\cos\varphi_1\right]^2 + \left[(a-a)\sin \varphi_1\right)^2
				}\leq\frac{K}{2m}\right)\\
				=&Pr_{\varphi_1\backsim [0,2\pi]}
				\left(\vert a\vert \cdot \vert \cos\varphi_1\vert \leq\frac{K}{2m}\right)
				\leq Pr_{\varphi_1\backsim [0,2\pi]}
				\left(\frac{K}{2}\cdot \vert \cos\varphi_1\vert \leq\frac{K}{2m}\right)\\
				=&Pr_{\varphi_1\backsim [0,2\pi]}
				\left(\vert \cos\varphi_1\vert \leq\frac{1}{m}\right)
				\end{align*}
				And we conclude by using Lemma \ref{lm:bounded}.1 again.
				
				\item{$Pr_{\substack{ \left\vert a \cdot e^{i\varphi_1} + b \cdot e^{-i\varphi_1} \right\vert>\frac{K}{2m}\\
				\left\vert \cos\theta\right\vert>\frac{1}{m^2}\\ 
				\left\vert \sin\theta\right\vert>\frac{1}{m}}}				
			\left(\left\vert 
			\cos\theta \left(a \cdot e^{i\varphi_1} + b \cdot e^{-i\varphi_1}\right)\right\vert- \left\vert \sin\theta \left(d \cdot e^{i\varphi_2}-c\cdot e^{-i\varphi_2}\right)
			\right\vert<\frac{K}{4m^3}\right)$}	
			\begin{align*}
			=&Pr_{\substack{ \left\vert a \cdot e^{i\varphi_1} + b \cdot e^{-i\varphi_1} \right\vert>\frac{K}{2m}\\
				\left\vert \cos\theta\right\vert>\frac{1}{m^2}\\ 
				\left\vert \sin\theta\right\vert>\frac{1}{m}}}				
			\left(\left\vert \left\vert
			cot(\theta)\right\vert\left\vert a \cdot e^{i\varphi_1} + b \cdot e^{-i\varphi_1}\right\vert- \left\vert d \cdot e^{i\varphi_2}-c\cdot e^{-i\varphi_2}
			\right\vert\right\vert<\frac{K}{4m^3\vert \sin\theta\vert}\right)\\
			\leq&Pr_{\substack{ \left\vert a \cdot e^{i\varphi_1} + b \cdot e^{-i\varphi_1} \right\vert>\frac{K}{2m}\\
				\left\vert \cos\theta\right\vert>\frac{1}{m^2}\\ 
				\left\vert \sin\theta\right\vert>\frac{1}{m}}}				
			\left(\left\vert \left\vert
			cot(\theta)\right\vert - \underbrace{\left\vert\frac{d \cdot e^{i\varphi_2}-c\cdot e^{-i\varphi_2}}{a \cdot e^{i\varphi_1} + b \cdot e^{-i\varphi_1}}\right\vert}_H\right\vert <\frac{K}{4m^2\left\vert a \cdot e^{i\varphi_1} + b \cdot e^{-i\varphi_1} \right\vert}\right)\\
			\leq&Pr_{\substack{ \left\vert a \cdot e^{i\varphi_1} + b \cdot e^{-i\varphi_1} \right\vert>\frac{K}{2m}\\
				\left\vert \cos\theta\right\vert>\frac{1}{m^2}\\ 
				\left\vert \sin\theta\right\vert>\frac{1}{m}}}	
				\left(\left\vert \left\vert
			cot(\theta)\right\vert - H\right\vert <\frac{1}{2m}\right)
			=Pr_{\substack{ \left\vert a \cdot e^{i\varphi_1} + b \cdot e^{-i\varphi_1} \right\vert>\frac{K}{2m}\\
				\left\vert \cos\theta\right\vert>\frac{1}{m^2}\\ 
				\left\vert \sin\theta\right\vert>\frac{1}{m}}}	
			\left(-\frac{1}{2m}< \left\vert	cot(\theta)\right\vert - H <\frac{1}{2m}\right)\\
			=&Pr_{\substack{ \left\vert a \cdot e^{i\varphi_1} + b \cdot e^{-i\varphi_1} \right\vert>\frac{K}{2m}\\
				\left\vert \cos\theta\right\vert>\frac{1}{m^2}\\ 
				\left\vert \sin\theta\right\vert>\frac{1}{m}}}	
			\left(H-\frac{1}{2m}< \left\vert cot(\theta)\right\vert <H+\frac{1}{2m}\right)		
			\end{align*}
			Again, due to the period and symmetry of the cotangent function, it suffices to evaluate this probability over the section $[0,\frac{\pi}{2}]$. We denote by $[x_1, x_2]\subset[0,\frac{\pi}{2}]$
 the section where $\vert \cos\theta\vert>\frac{1}{m^2},\vert \sin\theta\vert>\frac{1}{m}$, and remembering $m\geq2$, have $[x_1, x_2]\supset[\frac{\pi}{6},\frac{5\pi}{12}]$. As in the proof of Lemma \ref{lm:bounded}.1, we denote $\theta_{1}=cot^{-1}\left(\max\left(H-\frac{1}{2m}, 0\right)\right),\ \theta_{2}=cot^{-1}\left(H+\frac{1}{2m}\right),\ \theta_{1,2}\in [0,\frac{\pi}{2}]$ and as the cotangent function is monotonically decreasing in $[0,\frac{\pi}{2}]$ have:
				\[
				Pr_{\theta\backsim [x_1,x_2]}\left(H-\frac{1}{2m}< \left\vert cot(\theta)\right\vert <H+\frac{1}{2m}\right)=\frac{\theta_1-\theta_2}{x_1-x_2}
				\]
				and by the Mean Value Theorem and due to the fact the cotangent function's derivative is (in absolute value) at least 1, we can use The Mean Value Theorem on $\theta_{1,2}$ to achieve $\theta_1-\theta_2\leq\frac{H+\frac{1}{2m}-\left(H-\frac{1}{2m}\right)}{1}=\frac{1}{m}$

and so we get that the probability is less than $\frac{\theta_1-\theta_2}{\frac{5\pi}{12}-\frac{\pi}{6}}\leq\frac{4}{m\cdot\pi}<\frac{2}{m}$.
			\end{itemize}
					
			Summing the 4 probabilities, we get the induction base for Lemma \ref{lm:bounded}.	
	
			\item {\em Induction step:} Assuming correctness for ${n'}$, we prove for a matrix $\Delta$ on ${n'}+1$ qubits, 
assuming $\|\Delta\|\ge K$.
Choosing a basis for the ${n'}+1$ qubits where the $({n'}+1)$'th qubit 
is the most significant qubit, we have:
			\[
			u^{{n'}+1} = \begin{pmatrix}   
				u_1 & u_3 \\
				u_4 & u_2
			\end{pmatrix},
			u = \begin{pmatrix}   
				u_1 \cdot \begin{pmatrix} U \end{pmatrix}_{{n'}} & u_3 \cdot \begin{pmatrix} U \end{pmatrix}_{{n'}} \\
				u_4 \cdot \begin{pmatrix} U \end{pmatrix}_{{n'}} & u_2 \cdot \begin{pmatrix} U \end{pmatrix}_{{n'}} 
			\end{pmatrix},
			\Delta = \begin{pmatrix}   
				\begin{pmatrix} \Delta_1 \end{pmatrix}_{{n'}} & \begin{pmatrix} \Delta_3 \end{pmatrix}_{{n'}} \\
				\begin{pmatrix} \Delta_4 \end{pmatrix}_{{n'}} & \begin{pmatrix} \Delta_2 \end{pmatrix}_{{n'}}
			\end{pmatrix}
			\]
			Where $(U)_{{n'}}=u^1\otimes u^2 \otimes \cdots \otimes  u^{{n'}}$, and we denote the number of qubits each block operates on in subscript, for clarity. We now use the same notation as in the induction step 
in the proof of Lemma 
\ref{lm:basic}, 
and denote

\[
\Delta'=
\begin{pmatrix}   
				tr(\Delta_1\cdot U) & tr(\Delta_3\cdot U)\\
				tr(\Delta_4\cdot U) & tr(\Delta_2\cdot U)
			\end{pmatrix}
\]

			Again we can assume W.L.O.G that
 $\Vert\Delta_1\Vert_{Frob}\geq\frac{K}{2}$, and so by our 
induction step assumption applied on the $n'$ qubit matrix $\Delta_1$, we have:
			\begin{equation} \label{eq:62}
Pr_{u^1,\cdots, u^{n'} \backsim U(2)}\left(\Vert \Delta'\Vert_{Frob}<\frac{4K}{2\cdot \left(16m^3\right)^{n'}}\right)\leq Pr_{u^1,\cdots, u^{n'} \backsim U(2)}\left(\left\vert tr(\Delta_1 \cdot U)\right\vert<\frac{4K}{2\cdot \left(16m^3\right)^{n'}}\right)\leq \frac{5{n'}}{m}\end{equation}
where the first inequality is due to the fact that 
$|tr(\Delta_1\cdot U)|\le \Vert\Delta'\Vert_{Frob}$. By the induction base, 
applied to the one qubit matrix $\Delta'$, we also have 
			that if  $\Vert\Delta'\Vert_{Frob}\geq\frac{4K}{2\cdot \left(16m^3\right)^{n'}}$, 

			\begin{equation}\label{eq:induction base}
						Pr_{u^{{n'}+1}\backsim U(2)}\left(\left\vert tr(\Delta' \cdot u^{{n'}+1})\right\vert<\frac{4\cdot 4K}{16 m^3 \cdot 2\cdot \left(16m^3\right)^{n'}}\right)	
			\leq \frac{5}{m}. 
			\end{equation}
			
We now recall that as in 
the proof of Lemma \ref{lm:basic}, 
 Equation \refeq{eq:tr(delta*u)} we have
			\begin{equation}\label{eq:tr(delta*u) bounded}
			tr(\Delta\cdot u)=tr(\Delta'\cdot u^{{n'}+1})
			\end{equation}

			thus we can conclude by:
			\begin{align*}
			&Pr_{u^1,\cdots, u^{{n'}+1} \backsim U(2)}\left(\left\vert tr(\Delta \cdot u)\right\vert<\frac{4K}{\left(16m^3\right)^{{n'}+1}}\right)\\
			\leq& Pr_{u^1,\cdots, u^{n'} \backsim U(2)}\left(\Vert \Delta'\Vert_{Frob}<\frac{4K}{2\cdot \left(16m^3\right)^{n'}}\right) + 
			Pr_{\substack{u^{{n'}+1}\backsim U(2)\\ 
			\Vert \Delta'\Vert_{Frob}\geq\frac{4K}{2\cdot \left(16m^3\right)^{n'}}}}\left(\left\vert tr(\Delta' \cdot u^{{n'}+1})\right\vert<\frac{4K}{\left(16m^3\right)^{{n'}+1}}\right)\\
			\leq&\frac{5{n'}}{m}+
			Pr_{\substack{u^{{n'}+1}\backsim U(2)\\ 
			\Vert \Delta'\Vert_{Frob}\geq\frac{4K}{4\cdot \left(16m^3\right)^{n'}}}}\left(\left\vert tr(\Delta' \cdot u^{{n'}+1})\right\vert<\frac{4K}{8 \cdot m^3 \left(16m^3\right)^{n'}}\right)\leq \frac{5}{m}+\frac{5{n'}}{m}=\frac{5({n'}+1)}{m}
			\end{align*}
			When the first inequality follows from the total probability equation along with Equation \refeq{eq:tr(delta*u) bounded}, the second inequality follows from Inequality
\refeq{eq:62} along with the simple observation
that the probability to be smaller than a value - increases with that value, and the third inequality follows from Inequality \refeq{eq:induction base}.
		\end{itemize}

	\end{proof}

\clearpage
\section{Claim \ref{cl:bounded completeness} proof details}
\label{app:C5}

	Herein are the proofs of Claims \ref{cl:tr(A')-tr(A)}, \ref{cl:tr(M*h)-tr(M'*h)}, \ref{cl:tr(M'*h)-tr(M'*h')}, \ref{cl:tr(U'*0)-tr(U*0)}:
	
	\paragraph{Claim \ref{cl:tr(A')-tr(A)}}
	$\left\vert tr\left( A'_i\right)-tr\left( A_i\right)\right\vert\leq 2^n T\cdot \xi$
	\begin{proof}
	We have: 	
\[	
	\forall_{1\leq i\leq T}\left\Vert\tilde{u}_i-u_i\right\Vert_{L_2}\leq\xi\Rightarrow 
	\Vert A'_i-A_i \Vert \leq T\cdot\xi
\]
	Where the right inequality follows by a simple telescopic argument.\\
	
	Clearly, the difference between the traces of the two operators 
is bounded by the number of entries on their diagonal times their $L_2$ distance. The claim follows by the fact $A_i,A'_i$ have $2^n$ entries on their diagonal.
	
	\end{proof}

	\paragraph{Claim \ref{cl:tr(M*h)-tr(M'*h)}}
	$\left\vert tr\left( M_{i-1}\cdot g_i^{-1}\cdot u_i\right)-tr\left( M'_{i-1}\cdot g_i^{-1}\cdot u_i\right)\right\vert\leq 2^{2n'}\cdot 2^n T\cdot \xi$

	\begin{proof}
	We have:
	\begin{align*}
	\left\vert tr\left(M'_{i-1}\cdot g_i^{-1}\cdot u_i\right)-tr\left(M_{i-1}\cdot g_i^{-1}\cdot u_i\right)\right\vert 
	&=\left\vert tr\left(\left(M'_{i-1}-M_{i-1}\right)g_i^{-1}\cdot u_i\right)\right\vert
	=\left\vert \sum_{1\leq j,k \leq 2^{n'}}\left(M'_{i-1}-M_{i-1}\right)_{k,j}\cdot \left(g_i^{-1}\cdot u_i\right)_{j, k}\right\vert \\
	&\leq \sum_{1\leq j,k \leq 2^{n'}}\left\vert \left(M'_{i-1}-M_{i-1}\right)_{k, j}\right\vert
	\leq 2^{2n'}\cdot 2^n T\cdot \xi
	\end{align*}
	Where we used $\forall_{1\leq j,k \leq 2^{n'}}:\abs{ \left(g_i^{-1}\cdot u_i\right)_{j, k}} \leq 1$ for the first inequality, and a similar argument to that of Claim \ref{cl:tr(A')-tr(A)} for each entry of $\left(M'_{i-1}-M_{i-1}\right)$ in the second inequality (namely that $\abs{\left(M'_{i-1}-M_{i-1}\right)_{k,j}} \le \norm{M'_{i-1}-M_{i-1}} \le 2^n \norm{A'_{i-1}-A_{i-1}}$).

	\end{proof}
	
	\paragraph{Claim \ref{cl:tr(M'*h)-tr(M'*h')}}
	$\abs{tr\left( M'_{i-1}\cdot g^{-1}_i\cdot u_i\right)-tr\left( M'_{i-1}\cdot \widehat{g^{-1}}_i\cdot\widehat{u}_i\right)} \leq  2\cdot 2^{3n'}\cdot 2^n \cdot \xi$

	\begin{proof}
	We first note:
	\begin{align*}
	\max_{0\leq j,k\leq 2^{n'}}\abs{ \left( g^{-1}_i\cdot u_i - \widehat{g^{-1}}_i\cdot\widehat{u}_i\right)_{j,k}}
	&=\max_{0\leq j,k \leq 2^{n'}}\abs{\sum_{1\leq l\leq 2^{n'}}\left( g^{-1}_{i_{j,l}}\cdot u_{i_{l,k}} - \widehat{g^{-1}}_{i_{j,l}}\cdot\widehat{u}_{i_{l,k}}\right)}\\	
	&\leq 2^{n'}\cdot \max_{0\leq j,k,l \leq 2^{n'}}\abs{g^{-1}_{i_{j,k}}\cdot u_{i_{l,k}} - \widehat{g^{-1}}_{i_{j,k}}\cdot\widehat{u}_{i_{l,k}}}\\
	&\leq 2^{n'}\cdot \max_{0\leq j,k,l \leq 2^{n'}}\abs{g^{-1}_{i_{j,k}}\cdot u_{i_{l,k}} -g^{-1}_{i_{j,k}}\cdot \widehat{u}_{i_{l,k}} + g^{-1}_{i_{j,k}}\cdot \widehat{u}_{i_{l,k}} - \widehat{g^{-1}}_{i_{j,k}}\cdot\widehat{u}_{i_{l,k}}}\\
	&= 2^{n'}\cdot \max_{0\leq j,k,l \leq 2^{n'}}\abs{g^{-1}_{i_{j,k}}\left(u_{i_{l,k}} -\widehat{u}_{i_{l,k}}\right) + \widehat{u}_{i_{l,k}}\left(g^{-1}_{i_{j,k}} - \widehat{g^{-1}}_{i_{j,k}}\right)}\\
	&\leq 2^{n'}\cdot \left(\max_{0\leq j,k,l \leq 2^{n'}}\abs{ g^{-1}_{i_{j,k}}\left(u_{i_{l,k}} -\widehat{u}_{i_{l,k}} \right)} + \max_{0\leq j,k,l \leq 2^{n'}}\abs{ \widehat{u}_{i_{l,k}}\left(g^{-1}_{i_{j,k}} - \widehat{g^{-1}}_{i_{j,k}}\right)}\right)\\
	&\leq 2^{n'}\cdot \left(\max_{0\leq k,l \leq 2^{n'}}\abs{u_{i_{l,k}} -\widehat{u}_{i_{l,k}}} + \max_{0\leq j,k \leq 2^{n'}}\abs{g^{-1}_{i_{j,k}} - \widehat{g^{-1}}_{i_{j,k}}}\right)\\
	&\leq 2\cdot 2^{n'}\xi
	\end{align*}
	Where the fourth inequality follows from the fact that $u, g$ are unitaries. And so:
	\begin{align*}
	\left\vert tr\left( M'_{i-1}\cdot g^{-1}_i\cdot u_i\right)-tr\left( M'_{i-1}\cdot \widehat{g^{-1}}_i\cdot\widehat{u}_i\right) \right\vert
	&=\left\vert tr\left( M'_{i-1}\left( g^{-1}_i\cdot u_i - \widehat{g^{-1}}_i\cdot\widehat{u}_i\right)\right) \right\vert\\
	&=\left\vert \sum_{1\leq j,k \leq 2^{n'}}{M'_{i-1}}_{k,j}\cdot \left( g^{-1}_i\cdot u_i - \widehat{g^{-1}}_i\cdot\widehat{u}_i\right)_{j, k}\right\vert \\	
	&\leq 2^{2n'}\cdot 2^n \cdot \max_{0\leq j,k\leq 2^{n'}}\left\vert \left( g^{-1}_i\cdot u_i-\widehat{g^{-1}}_i\cdot\widehat{u}_i\right)_{j,k}\right\vert
	\leq 2\cdot 2^{3n'}\cdot 2^n \cdot \xi
	\end{align*}
	Where we used $\displaystyle\max_{0\leq j,k\leq 2^{n'}}\left\vert M_{i-1_{j,k}}\right\vert \leq 2^n$ (as explained in the protocol description) for the first inequality.
	\end{proof}
	
	\paragraph{Claim \ref{cl:tr(U'*0)-tr(U*0)}}
	$\left\vert tr(\widehat{U}_T\cdot \widehat{U}_{T-1}\cdots \widehat{U}_1 \cdot \vert0^n\rangle \langle 0^n \vert) - tr(U_T\cdot U_{T-1}\cdots U_1 \cdot \vert0^n\rangle \langle 0^n \vert)\right\vert\leq 2^{n}T\cdot\xi$
	
	\begin{proof}
	\begin{align*}
	\left\vert tr(\widehat{U}_T\cdots \widehat{U}_1 \cdot \vert0^n\rangle \langle 0^n \vert) - tr(U_T\cdots U_1 \cdot \vert0^n\rangle \langle 0^n \vert)\right\vert
	&=\left\vert\langle 0^n \vert \widehat{U}_T\cdots \widehat{U}_1- U_T\cdots U_1 \vert0^n\rangle\right\vert \\
	&\leq \left\Vert \widehat{U}_T\cdots \widehat{U}_1- U_T\cdots U_1 \right\Vert_{L_2}
	\leq T \displaystyle\max_{0\leq i\leq T}\left\Vert \widehat{U}_i - U_i\right\Vert_{Frob}
	\leq 2^{n}T\cdot\xi
	\end{align*}
Where we used a telescopic argument together with the fact that the fact the operator norm is bounded from above by the Frobenius norm for the second inequality.
	\end{proof}

\clearpage
	\section{Claim \ref{cl:bounded soundness} proof details}
	\label{app:C6}
Herein is the proof of Claim \ref{cl:prob. dist}.

\paragraph{Claim \ref{cl:prob. dist}}
let $\Delta$ an operator on $n'$ qubits:
\[
Pr_{u\backsim \widehat{D}(n')}\Big(\left\vert tr\left(\Delta\cdot u \right)\right\vert < \delta-2^{n'}n'\cdot6\cdot\xi\cdot\Vert\Delta\Vert_{Frob} \Big)
\leq Pr_{u\backsim D(n')}\Big(\left\vert tr\left(\Delta\cdot u \right)\right\vert < \delta \Big)+3n'\frac{\xi}{2\pi}
\]

\begin{proof}
To prove the claim, we remember that choosing a unitary $u\backsim U(2)$ is equivalent to choosing $\theta,\varphi_1,\varphi_2\backsim [0,2\pi)$ by defining:
\[
	u=\begin{pmatrix}   
	\cos\theta\cdot e^{i\varphi_1} & \sin\theta\cdot e^{i\varphi_2} \\
	-\sin\theta\cdot e^{-i\varphi_2} & \cos\theta\cdot e^{-i\varphi_1}
	\end{pmatrix}
\]
And $u\backsim U(2)\Leftrightarrow \theta,\varphi_1,\varphi_2\backsim [0,2\pi)$.\\

Choosing a unitary $u\backsim D(n')$ is then equivalent to choosing 
$\theta^1,\varphi^1_1,\varphi^1_2,\cdots,\theta^{n'},\varphi^{n'}_1,\varphi^{n'}_2\backsim [0,2\pi)$ for $u^1,\cdots,u^{n'}$, with $u=u^1\otimes\cdots\otimes u^{n'}$.\\

We also note that for any $\kappa\geq2\pi$, and any event $F$:
\[
Pr_{x\backsim[0,\kappa)}\Big(F\Big)
\leq Pr_{x\backsim[0,\kappa)}\Big(x\in[2\pi,\kappa)\Big)
+Pr_{\substack{
x \backsim [0,\kappa)\\
x \notin [2\pi,\kappa)
}}\Big(F\Big)=
\frac{\kappa-2\pi}{\kappa}+Pr_{x\backsim[0,2\pi)}\Big(F\Big)
\]
By applying the previous argument $3{n'}$ times we get:
\[
Pr_{\theta^1,\varphi^1_1,\varphi^1_2,\cdots,\theta^{n'},\varphi^{n'}_1,\varphi^{n'}_2\backsim[0,\kappa)}\Big(F\Big)
\leq 3{n'}\frac{\kappa-2\pi}{\kappa}+Pr_{\theta^1,\varphi^1_1,\varphi^1_2,\cdots,\theta^{n'},\varphi^{n'}_1,\varphi^{n'}_2\backsim[0,2\pi)}\Big(F\Big)
=3{n'}\frac{\kappa-2\pi}{\kappa} + Pr_{u\backsim D({n'})}\Big(F\Big)
\]
Setting $\displaystyle\kappa=\min_{\substack{i\in \mathbb{N}\\ i\cdot\xi\geq2\pi}}\left(i\cdot\xi\right)$, we have $\kappa-2\pi<\xi$ and so Claim \ref{cl:prob. dist} can be proven by showing:
\begin{equation}\label{temp3}
Pr_{u\backsim \widehat{D}({n'})}\Big(\left\vert tr\left(\Delta\cdot u \right)\right\vert < \delta-2^{{n'}}{n'}\cdot6\cdot\xi\cdot\Vert\Delta\Vert_{Frob} \Big)
\leq Pr_{\theta^1,\varphi^1_1,\varphi^1_2,\cdots,\theta^{n'},\varphi^{n'}_1,\varphi^{n'}_2\backsim[0,\kappa)}\Big(\left\vert tr\left(\Delta\cdot u \right)\right\vert < \delta \Big)
\end{equation}

Now, let us denote the rounding down of any value $v$ to within accuracy $\xi$ by $\ddot{v}$ (formally, $\ddot{v} = \displaystyle \max_{\substack{i\in \mathbb{N}\\ i\cdot\xi\leq v}}i\cdot\xi$). Assuming we generate a unitary $u$ by choosing $\theta^1,\varphi^1_1,\varphi^1_2,\cdots,\theta^{n'},\varphi^{n'}_1,\varphi^{n'}_2\backsim[0,\kappa)$, let us abuse notation and use $\ddot{u}$ to denote the unitary which is defined by $\ddot{\theta}^1,\ddot{\varphi}^1_1,\ddot{\varphi}^1_2,\cdots,\ddot{\theta}^{n'},\ddot{\varphi}^{n'}_1,\ddot{\varphi}^{n'}_2$. It is 
straightforward to see that $\ddot{u}\backsim\widehat{D}({n'})$, so (\ref{temp3}) will follow from proving:
\begin{equation}\label{temp4}
\vert tr(\Delta\cdot u)\vert=\delta\Rightarrow\vert tr(\Delta\cdot\ddot{u})\vert\geq\delta - 2^{{n'}}{n'}\cdot6\cdot\xi\cdot\Vert\Delta\Vert_{Frob}
\end{equation}
We note:
\[
\delta
=\vert tr(\Delta\cdot u)\vert
=\vert tr(\Delta(u-\ddot{u}))+tr(\Delta\cdot\ddot{u})\vert
\leq\vert tr(\Delta(u-\ddot{u}))\vert+\vert tr(\Delta\cdot\ddot{u})\vert
\]
\begin{equation}\label{temp5}
\Rightarrow\vert tr(\Delta\cdot\ddot{u})\vert
\geq\delta-\vert tr(\Delta(u-\ddot{u}))\vert
\end{equation}

Also, similarly to Claim \ref{cl:tr(A')-tr(A)}
\begin{equation}\label{temp6}
\vert tr(\Delta(u-\ddot{u}))\vert
\leq 2^{n'}\Vert\Delta(u-\ddot{u})\Vert 
\leq 2^{n'}\Vert\Delta\Vert\Vert u-\ddot{u}\Vert
\leq 2^{n'}\Vert\Delta\Vert_{Frob}\cdot {n'} \cdot \max_{1\leq i\leq {n'}} \Vert u^i-\ddot{u}^i\Vert
\end{equation}
And:
\begin{equation}\label{temp7}
\max_{1\leq i\leq {n'}} \Vert u^i-\ddot{u}^i\Vert
\leq\max_{1\leq i\leq {n'}} \Vert u^i-\ddot{u}^i\Vert_{Frob}
\leq\max_{0\leq\theta,\varphi<\kappa} 2\cdot\vert \cos\theta\cdot e^{i\varphi}-\cos\ddot{\theta}\cdot e^{i\ddot{\varphi}} \vert
\end{equation}

Ultimately, we have:
\[
\max_{0\leq\theta,\varphi<\kappa} \vert \cos\theta\cdot e^{i\varphi}-\cos\ddot{\theta}\cdot e^{i\ddot{\varphi}} \vert
=\max_{0\leq\theta,\varphi<\kappa} \sqrt{\left(\cos\theta \cos\varphi-\cos\ddot{\theta}\cos\ddot{\varphi}\right)^2+\left(\cos\theta \sin\varphi-\cos\ddot{\theta}\sin\ddot{\varphi}\right)^2} \\
\]
\begin{equation}\label{temp8}
\Rightarrow \max_{0\leq\theta,\varphi<\kappa} \vert \cos\theta\cdot e^{i\varphi}-\cos\ddot{\theta}\cdot e^{i\ddot{\varphi}} \vert
\leq \max_{0\leq\theta<\kappa} 2\sqrt{2}\vert\cos\theta-\cos\ddot{\theta}\vert
\leq \max_{0\leq\theta<\kappa} 2\sqrt{2}\vert\theta-\ddot{\theta}\vert
\leq 2\sqrt{2}\xi
\end{equation}

We conclude the proof by noting that Inequality \refeq{temp4} follows from the Inequalities \refeq{temp5}, \refeq{temp6}, \refeq{temp7}, \refeq{temp8}. 
\end{proof}

\end{document}